\PassOptionsToPackage{unicode}{hyperref}
\PassOptionsToPackage{hyphens}{url}
\PassOptionsToPackage{dvipsnames,svgnames,x11names}{xcolor}
\documentclass[12pt]{article}

\usepackage{amsmath,amssymb}
\usepackage{enumerate}
\usepackage{iftex}
\usepackage{bbm}
\usepackage{tikz}
\usetikzlibrary{decorations.pathreplacing}
\ifPDFTeX
  \usepackage[T1]{fontenc}
  \usepackage[utf8]{inputenc}
  \usepackage{textcomp} 
\else 
  \usepackage{unicode-math}
  \defaultfontfeatures{Scale=MatchLowercase}
  \defaultfontfeatures[\rmfamily]{Ligatures=TeX,Scale=1}
\fi
\usepackage{lmodern}
\ifPDFTeX\else  
\fi
\IfFileExists{upquote.sty}{\usepackage{upquote}}{}
\IfFileExists{microtype.sty}{
  \usepackage[]{microtype}
  \UseMicrotypeSet[protrusion]{basicmath} 
}{}
\makeatletter
\@ifundefined{KOMAClassName}{
  \IfFileExists{parskip.sty}{%
    \usepackage{parskip}
  }{
    \setlength{\parindent}{0pt}
    \setlength{\parskip}{6pt plus 2pt minus 1pt}}
}{
  \KOMAoptions{parskip=half}}
\makeatother
\usepackage{xcolor}
\makeatletter
\ifx\paragraph\undefined\else
  \let\oldparagraph\paragraph
  \renewcommand{\paragraph}{
    \@ifstar
      \xxxParagraphStar
      \xxxParagraphNoStar
  }
  \newcommand{\xxxParagraphStar}[1]{\oldparagraph*{#1}\mbox{}}
  \newcommand{\xxxParagraphNoStar}[1]{\oldparagraph{#1}\mbox{}}
\fi
\ifx\subparagraph\undefined\else
  \let\oldsubparagraph\subparagraph
  \renewcommand{\subparagraph}{
    \@ifstar
      \xxxSubParagraphStar
      \xxxSubParagraphNoStar
  }
  \newcommand{\xxxSubParagraphStar}[1]{\oldsubparagraph*{#1}\mbox{}}
  \newcommand{\xxxSubParagraphNoStar}[1]{\oldsubparagraph{#1}\mbox{}}
\fi
\makeatother

\usepackage{longtable,booktabs,array}
\usepackage{calc} 
\usepackage{etoolbox}
\makeatletter
\patchcmd\longtable{\par}{\if@noskipsec\mbox{}\fi\par}{}{}
\makeatother
\IfFileExists{footnotehyper.sty}{\usepackage{footnotehyper}}{\usepackage{footnote}}
\makesavenoteenv{longtable}
\usepackage{graphicx}
\makeatletter
\def\maxwidth{\ifdim\Gin@nat@width>\linewidth\linewidth\else\Gin@nat@width\fi}
\def\maxheight{\ifdim\Gin@nat@height>\textheight\textheight\else\Gin@nat@height\fi}
\makeatother
\setkeys{Gin}{width=\maxwidth,height=\maxheight,keepaspectratio}
\makeatletter
\def\fps@figure{htbp}
\makeatother

\makeatletter
\@ifpackageloaded{caption}{}{\usepackage{caption}}
\AtBeginDocument{%
\ifdefined\contentsname
  \renewcommand*\contentsname{Table of contents}
\else
  \newcommand\contentsname{Table of contents}
\fi
\ifdefined\listfigurename
  \renewcommand*\listfigurename{List of Figures}
\else
  \newcommand\listfigurename{List of Figures}
\fi
\ifdefined\listtablename
  \renewcommand*\listtablename{List of Tables}
\else
  \newcommand\listtablename{List of Tables}
\fi
\ifdefined\figurename
  \renewcommand*\figurename{Figure}
\else
  \newcommand\figurename{Figure}
\fi
\ifdefined\tablename
  \renewcommand*\tablename{Table}
\else
  \newcommand\tablename{Table}
\fi
}
\@ifpackageloaded{float}{}{\usepackage{float}}
\floatstyle{ruled}
\@ifundefined{c@chapter}{\newfloat{codelisting}{h}{lop}}{\newfloat{codelisting}{h}{lop}[chapter]}
\floatname{codelisting}{Listing}

\makeatother
\makeatletter
\@ifpackageloaded{caption}{}{\usepackage{caption}}
\@ifpackageloaded{subcaption}{}{\usepackage{subcaption}}
\makeatother

\ifLuaTeX
  \usepackage{selnolig}  
\fi
\usepackage[]{natbib}
\usepackage{bookmark}

\IfFileExists{xurl.sty}{\usepackage{xurl}}{} 
\hypersetup{
  pdftitle={Title},
  pdfauthor={Author 1; Author 2},
  pdfkeywords={3 to 6 keywords, that do not appear in the title},
  colorlinks=true,
  linkcolor={blue},
  filecolor={Maroon},
  citecolor={Blue},
  urlcolor={Blue},
  pdfcreator={LaTeX via pandoc}}

\newtheorem{theorem}{Theorem}
\newtheorem{lemma}{Lemma}
\newtheorem{remark}{Remark}
\newtheorem{proposition}{Proposition}

\newtheorem{defin}{Definition}
\newtheorem{condition}{Condition}

\def\bm{\boldsymbol}

\def\beq{\begin{equation}}
\def\eeq{\end{equation}}
\def\beqs{\begin{equation*}}
\def\eeqs{\end{equation*}}
\def\beqr{\begin{eqnarray}}
\def\eeqr{\end{eqnarray}}
\def\beqrs{\begin{eqnarray*}}
\def\eeqrs{\end{eqnarray*}}
\def\bet{\begin{theorem}}
\def\eet{\end{theorem}}
\def\bel{\begin{lemma}}
\def\eel{\end{lemma}}
\def\bep{\begin{proposition}}
\def\eep{\end{proposition}}
\def\bg{\begin{figure}[tbph]\begin{center}}
\def\eg{\end{center}\end{figure}}

\def\bc{\begin{center}}
\def\ec{\end{center}}

\newcommand{\anon}{1}

\begin{document}

\def\spacingset#1{\renewcommand{\baselinestretch}%
{#1}\small\normalsize} \spacingset{1}


\if1\anon
{
  \title{\bf Local spectral clustering for heterogeneous clustering structures}
  \author{
    Yuanxing Chen \\
    School of Economics and Management, Fuzhou University\\
    and\\
    Qingzhao Zhang \\
    School of Economics, Xiamen University\\
    The Wang Yanan Institute for Studies in Economics, Xiamen University\\
    and\\
    Yuhong Yang\\
    Yau Mathematical Sciences Center, Tsinghua University\\
    Beijing Institute of Mathematical Sciences and Applications
    }
  \maketitle
} \fi

\if0\anon
{
  \bigskip
  \bigskip
  \bigskip
  \begin{center}
    {\LARGE\bf Fused Spatial Latent Block Models for Co-Clustering}
\end{center}
  \medskip
} 
\fi

\bigskip
\begin{abstract}
Classical clustering methods typically assume that all informative features support a single latent partition of the observations. This assumption can be overly restrictive for modern high-dimensional data, where different subsets of features may encode distinct notions of similarity and induce heterogeneous sample partitions, while some features may contain no meaningful clustering information. We develop a frequentist framework for local clustering that simultaneously identifies feature groups and estimates the sample clustering structure associated with each group. Our approach represents each sample partition by a label-invariant clustering matrix and groups features according to their shared clustering structures, thereby reformulating local clustering as a feature-grouping, or clustering-of-clusterings, problem. Under a heterogeneous sub-Gaussian mixture model, we construct feature-specific Gaussian-kernel similarity matrices and propose a local spectral clustering procedure based on a clustering-matrix optimization criterion. The proposed method avoids explicit likelihood specification and Bayesian posterior computation, accommodates heterogeneous feature distributions, and permits the presence of non-informative features. Extensive simulations and applications further demonstrate the practical utility and superiority of the proposed
approach.
\end{abstract}

\noindent%
\vfill

\newpage

\spacingset{1.8} 
\section{Introduction}\label{sec1}

Clustering is one of the most fundamental tasks in statistical learning, aiming to partition a collection of observations into internally homogeneous and mutually heterogeneous groups based on certain similarities in the observed features. 
Classical approaches, including $K$-means clustering \citep{jain2010data}, hierarchical clustering \citep{murtagh2012algorithms}, model-based clustering \citep{gormley2023model}, and spectral clustering \citep{von2007tutorial}, differ substantially in their formulations and computational strategies. Nevertheless, they are generally built upon a common structural assumption: a dataset is characterized by a single latent partition of the observations. Under this global clustering paradigm, all informative features are expected to contribute, possibly with different strengths, to the recovery of the same underlying clustering structure.

Although the assumption of a single global partition is natural and desirable in many classical applications, it can be overly restrictive for modern high-dimensional datasets \citep{witten2010framework,lee2013nonparametric}.
In such cases, features are often highly heterogeneous in their scientific meanings, dependence structures, and sources of variation, and only subsets of features may carry meaningful clustering information
relevant to particular clustering patterns \citep{raftery2006variable,liu2023clustering}.
Consequently, different subsets of features may encode different notions of similarity among the same observations, thereby inducing distinct sample partitions \citep{lee2013nonparametric,xu2013nonparametric}. 
Rather than being characterized by a unique global clustering structure, a high-dimensional dataset may therefore contain multiple clustering structures that coexist across different subsets of features.

This phenomenon is particularly common in modern omics and also epidemiological studies, where distinct subsets of features often correspond to different biological processes or disease types and, therefore, induce different partitions of the same observations. 
The success of gene set enrichment analysis \citep{subramanian2005gene} demonstrated that  biologically meaningful signals frequently arise from the coordinated behavior of groups of genes, such as pathways or functional modules, rather than from isolated individual genes. 
Genes involved in different biological pathways may distinguish observations according to various disease subtypes, cellular states, or molecular mechanisms. As a result, one subset of genes may separate observations according to a particular biological process, whereas another subset may induce a substantially different partition of the same observations \citep{xie2019time}.
In proteomic studies, different protein modules reflect heterogeneous biological mechanisms and generate different patterns of sample similarity \citep{lee2013nonparametric,hao2021integrated}.
Similar forms of structural heterogeneity arise in electronic health records (EHR), where different sets of clinical indicators reflect distinct latent diseases \citep{ni2020bayesian}.
More broadly, whenever high-dimensional features arise from multiple functional systems, measurement platforms, or data-generating mechanisms, it is generally unrealistic to expect all features to support a common sample partition.

These examples motivate a more flexible framework for clustering structures, which is referred to as local clustering in \cite{lee2013nonparametric}. Instead of asking for a single partition that summarizes the entire feature space, local clustering seeks to determine which features support the same partition and jointly identify a non-overlapping feature partition while recovering the corresponding non-overlapping clustering structure within each feature group. 
Importantly, the feature partition in local clustering is unknown a priori and depends on the sample clustering structures it induces. Therefore, separate clustering analyses cannot be conducted in advance for different feature groups.

The need to identify localized clustering patterns has motivated an extensive literature on biclustering and related methods. Biclustering seeks subsets of observations and features that exhibit coherent local patterns within a data matrix \citep{cheng2000biclustering,lazzeroni2002plaid,chi2017convex,flynn2020profile,zhong2022biclustering}. 
These biclustering methods simultaneously cluster features and observations by identifying submatrices with coherent numerical patterns, rather than grouping features according to the sample partitions they induce.
Consequently, biclustering does not always directly address the problem of recovering multiple complete partitions of the same set of observations.
Bayesian local clustering methods provide a more direct formulation of this problem by allowing different subsets of features to induce different partitions of a common set of observations. An influential work is the nonparametric Bayesian local clustering model proposed by \cite{lee2013nonparametric}, which formalized the idea that features can be grouped according to a shared partition of the observations. Related hierarchical and nonparametric constructions have further expanded the flexibility of this approach \citep{rodriguez2008nested,dombowsky2025bayesian}. These developments offer a powerful alternative to global clustering and demonstrate the practical importance of explicitly modeling heterogeneous clustering structures.


Despite their flexibility and empirical success, existing local clustering methods face several limitations. First, most available approaches are model-based and rely on hierarchical Bayesian formulations, latent allocation features, or elaborate likelihood specifications \citep{rodriguez2008nested,lee2013nonparametric,dombowsky2025bayesian}. 
Statistical inference may therefore require computationally intensive posterior sampling procedures. 
Second, existing methods typically describe local clustering structures indirectly through latent parameters in a probabilistic model \citep{tan2014sparse,chi2017convex}. The partitions themselves are not always treated as explicit statistical objects in the formulation of the estimation problem. This makes it difficult to directly compare clustering structures across features or to characterize the statistical separation required to distinguish one feature group from another.
Third, a complete theoretical analysis must account for two interconnected sources of uncertainty. One needs to simultaneously estimate the grouping of features and the sample partition associated with each feature group. Errors in feature grouping may propagate into the estimation of sample clusters, while inaccurate sample partitions may, in turn, obscure the distinction between feature groups. Establishing recovery guarantees therefore requires a unified analysis of these two layers of latent structure, which is basically absent to the best of our knowledge.

In this paper, we develop a frequentist framework for local clustering based on an explicit representation of heterogeneous clustering structures. 
Our starting point is that the central object in local clustering is the partition of observations rather than a collection of latent model parameters. 
We represent each sample partition by its clustering matrix, whose entries record whether pairs of observations belong to the same cluster. 
This representation is invariant to permutations of cluster labels and therefore provides a natural object for comparing clustering structures across features.
Under this formulation, features belong to the same feature group when they induce the same clustering matrix. Local clustering can consequently be reformulated as a feature grouping problem in which features are grouped according to their shared sample partitions. Equivalently, the problem may be viewed as a clustering-of-clusterings problem: each feature contains information about a sample clustering structure, and the objective is to identify groups of features that generate identical or sufficiently similar partitions. This representation separates the structural target of interest from any particular likelihood specification and provides a direct basis for estimation and theoretical analysis.

To formalize this idea, we introduce a heterogeneous sub-Gaussian mixture model under which different feature groups induce different partitions of the observations. 
Based on feature-specific Gaussian kernel similarity matrices, we propose a local spectral clustering (LSC) framework that jointly estimates feature groups and their associated clustering structures through a clustering-matrix optimization criterion. 
As a result, within each feature group, the corresponding features share a common sample clustering structure, although their cluster centers, noise levels, and marginal distributions may vary. The formulation also permits non-informative features that are unrelated to any meaningful sample partition, which is important in high-dimensional applications where only a subset of the measured features may contain clustering information.

A key advantage of the proposed framework is that it enables a comprehensive theoretical analysis of local clustering. We first derive non-asymptotic estimation error bounds for an oracle estimator that knows the true feature grouping structure. 
We then show that the oracle estimator can be recovered as a local maximizer of the proposed objective function with high probability, leading to the asymptotic recovery of the true feature grouping structure. Finally, by combining spectral perturbation arguments with approximate $K$-means theory, we derive uniform upper bounds for the misclassification errors of the estimated sample partitions. 
These results simultaneously control the clustering errors across multiple feature groups and reveal how uncertainty in the estimated similarity matrices propagates through spectral embedding and subsequent cluster assignment. Together, the theoretical results provide recovery guarantees for both levels of the local clustering problem: the grouping of features and the clustering of observations within each feature group.
To the best of our knowledge, these results provide the first frequentist understanding of the local clustering under heterogeneous feature-specific clustering structures.

Our major contributions are fourfold:
\begin{itemize}
    \item First, we introduce a new statistical formulation of local clustering based on shared sample partitions. Unlike existing approaches that define local clusters through latent model parameters, our framework treats clustering matrices as the primary statistical objects and formulates local clustering as a feature grouping problem.
    \item Second, we propose a computationally efficient optimization framework that simultaneously identifies informative feature groups and estimates their associated clustering structures. The proposed method does not require likelihood maximization or Bayesian posterior computation and is scalable to high-dimensional settings.
    \item Third, we establish a unified theoretical framework for local clustering, which characterizes the recovery of group-specific clustering structures, the asymptotic recovery of feature groups, and the consistency of local spectral clustering under heterogeneous sub-Gaussian mixture models. The results provide explicit signal-to-noise conditions under which heterogeneous clustering structures can be successfully recovered.
    \item Fourth, we apply our local spectral clustering method to a real data and it reveals that a group of AML patients may benefit from one available treatment over the other, which seems to be an interesting finding.
\end{itemize}

The remainder of the paper is organized as follows. Section 2 introduces the model formulation, clustering-matrix representation, and local spectral clustering procedure. Section 3 presents the theoretical properties of the proposed method, including oracle recovery, feature-group recovery, and clustering consistency. Section 4 investigates the finite-sample performance through simulation studies. Section 5 illustrates the proposed framework using real data examples. Section 6 concludes with a discussion.

\section{Methodology}\label{sec2}
\subsection{Notation}
For any positive integer $n$, define $[n]=\{1,\ldots,n\}$.
For two sets $\mathcal A$ and $\mathcal B$, denote by $\mathcal A\setminus \mathcal B$ the set difference $\{x:x\in \mathcal A,x\notin \mathcal B\}$.
Given an index set $\mathcal S\subseteq[n]$, let $|\mathcal S|$ be the cardinality of $\mathcal S$.
For any vector $\mathbf a=(a_1,\dots,a_n)^\top$, let $\mathbf a_{\mathcal S}$ denote a vector of length $|\mathcal S|$ whose components are $\{a_i:i\in\mathcal S\}$. Meanwhile, its $\ell_2$ norm is defined by $\|\mathbf a\|_2=\sqrt{\sum_{j=1}^pa_j^2}$. 
For any matrix $\mathbf A\in\mathbb R^{n\times n}$, let $\mathrm{Tr}(\mathbf A)$ be its trace and $A_{ij}$ be its $(i,j)$-th entry. In addition, we define its $\ell_2$ norm, $\ell_1$ norm, sup-norm, and Frobenius norm as $\|\mathbf A\|_2=\sup_{\mathbf v\in\mathbb R^n,\|\mathbf v\|_2=1}\|\mathbf A\mathbf v\|_2$, $\|\mathbf A\|_1=\sum_{i=1}^n\sum_{j=1}^n|A_{ij}|$, $\|\mathbf A\|_{\max}=\max_{i,j\in[n]}|A_{ij}|$, and $\|\mathbf A\|_F=\sqrt{\sum_{i=1}^n\sum_{j=1}^nA_{ij}^2}$, respectively. For another matrix $\mathbf B\in\mathbb R^{n\times n}$, the inner product between $\mathbf A$ and $\mathbf B$ is denoted by $\langle\mathbf A,\mathbf B\rangle:=\mathrm{Tr}(\mathbf A^\top\mathbf B)=\sum_{i=1}^n\sum_{j=1}^nA_{ij}B_{ij}$. We represent the $n\times n$ matrix of all ones by $\mathbf E_n$ and denote $\mathbf I_s$ as the $s$-dimensional identity matrix. 
For two sequences of real numbers $\{a_n\}$ and $\{b_n\}$, we use $a_n\gg b_n$ (or $b_n\ll a_n$) if $b_n/a_n=o(1)$, and $a_n\asymp b_n$ if $a_n$ is of the same order as $b_n$.
For a centered random variable $X\in\mathbb R$, we say $X$ is a sub-Gaussian random variable with sub-Gaussian norm $\sigma_x>0$ if $\mathbb E(\exp\{\lambda X\})\le\exp\{\sigma_x^2\lambda^2/2\}$ for any $\lambda\in\mathbb R$.

\subsection{Model setting}

Let $\mathbf X=(\mathbf x_1,\dots,\mathbf x_n)^\top \in \mathbb{R}^{n\times p}$ denote the observed data matrix, where $\mathbf x_i=(x_{i1},\dots,x_{ip})^\top$.
We assume that the $p$ features can be partitioned into $G$ informative groups and one non-informative group. For each informative group $g\in[G]$, the features within the group induce a common clustering structure on the samples, with $M_g\ge 2$ clusters. In contrast, the non-informative group does not exhibit any clustering structure, corresponding to a single cluster.
Let $\mathcal A^*\subset[p]$ index the informative features, and let ${\mathcal A_1^*,\dots,\mathcal A_G^*}$ be a partition of $\mathcal A^*$, where each $\mathcal A_g^*$ corresponds to the feature indices in the $g$-th informative group. The complement $\mathcal A_{G+1}^*=[p]\setminus \mathcal A^*$ represents the non-informative feature group.

The observations are generated from $G$ sub-Gaussian mixture models (SGMMs) and one sub-Gaussian model (SGM). For each $g\in [G]$, the $g$-th SGMM consists of $M_g$ sub-Gaussian components, denoted by ${\mathcal D_{m}^{(g)}: m\in [M_g]}$.
For the $g$-th SGMM, let $\bm\pi^{(g)}=(\pi_{1}^{(g)},\dots,\pi_{M_g}^{(g)})^\top$ be the mixing probability vector, where $\sum_{m=1}^{M_g}\pi_{m}^{(g)}=1$ and $\pi_{m}^{(g)}>0$ for all $m\in [M_g]$. Let $\bm\mu_{m}^{(g)}=(\mu_{m,1}^{(g)},\dots,\mu_{m,s_g}^{(g)})^\top \in \mathbb{R}^{s_g}$ denote the mean vector of the $m$-th component, where $s_g = |\mathcal A_g^*|$.
For each observation $i\in [n]$, let $\psi_i^{(g)}\in [M_g]$ denote its latent cluster label under the $g$-th SGMM, which is drawn from a multinomial distribution with parameter $\bm\pi^{(g)}$. Accordingly, for each $g\in[G]$, define the group-specific sample partition $\mathcal G_{g,1}^*,\dots,\mathcal G_{g,M_g}^*$ of $[n]$, where $\mathcal G_{g,m}^* = \{i\in [n]: \psi_i^{(g)} = m\}$ and $n_m^{(g)} = |\mathcal G_{g,m}^*|$ is the size of the $m$-th cluster.
The corresponding true membership matrix is given by $\mathbf Z_g^* \in \{0,1\}^{n\times M_g}$, where $Z_{g,im}^* = 1$ if and only if $i\in \mathcal G_{g,m}^*$. Without loss of generality, assume that observations from the same cluster are indexed together. Then, under the $g$-th informative feature group, the associated clustering matrix $\mathbf P_g^* = \mathbf Z_g^* \mathbf Z_g^{*\top}$ is block diagonal, with $P_{g,ii'}^* = 1$ if $i$ and $i'$ belong to the same cluster, and $0$ otherwise.
In contrast, for the non-informative feature group (i.e., the $(G+1)$-th group), all observations are homogeneous. In this case, $\mathbf Z_{G+1}^*$ reduces to an $n$-dimensional vector of ones, and consequently $\mathbf P_{G+1}^* = \mathbf E_n$.

\begin{figure}[H]
\centering
\includegraphics[height =8.0cm]{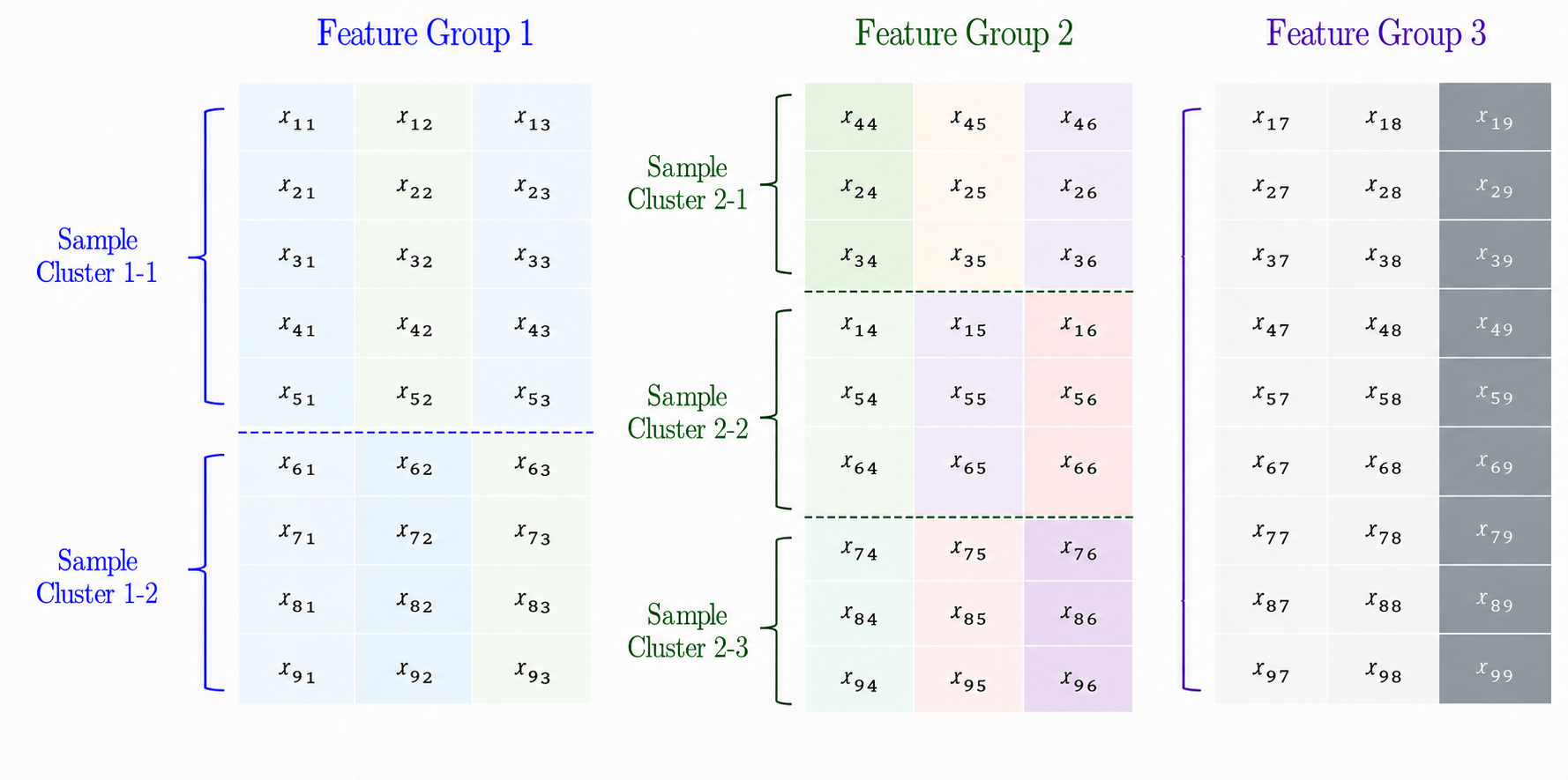}
\caption{{\color{red} Data generating structure with group heterogeneity.
}}
\label{clustermap}
\end{figure}

Figure \ref{clustermap} illustrates the group-heterogeneous structure considered in this paper. The features are partitioned into several latent feature groups, where features within the same group induce a common partition of the samples, while different groups correspond to distinct clustering structures. 
For example, Features 1–3 partition the samples into two clusters, whereas Features 4–6 partition the same samples into three clusters under a different latent mechanism. In contrast, Features 7–9 do not exhibit any meaningful clustering pattern and can be viewed as non-informative features. This example highlights a key challenge in modern high-dimensional data analysis: different subsets of features may capture different underlying mechanisms and therefore generate heterogeneous sample partitions. Consequently, the objective of local clustering is not to recover a single global clustering structure, but rather to identify groups of features that share common sample partitions and to recover the corresponding local clustering structures.

Based on the above SGMM formulation, for any $g\in[G]$, conditional on $Z_{g,im}^* = 1$, the $i$-th observation restricted to the $g$-th feature group, $\mathbf x_{i,\mathcal A_g^*}$, is generated from the $m$-th component distribution $\mathcal D_m^{(g)}$, i.e.,
\[
\mathbf x_{i,\mathcal A_g^*}:=\bm\mu_{m}^{(g)}+\bm\epsilon_{i,\mathcal A_g^*}.
\]
Here, $\bm\epsilon_i = (\epsilon_{i1},\dots,\epsilon_{ip})^\top$, where $\{\epsilon_{ij}\}_{j=1}^p$ are mean-zero sub-Gaussian random variables with sub-Gaussian norm $\sigma$.
For the non-informative feature group, the observations are generated from a single sub-Gaussian distribution. Specifically,
\[
\mathbf x_{i,\mathcal A_{G+1}^*}:=\bm\mu^{(G+1)}+\bm\epsilon_{i,\mathcal A_{G+1}^*},
\]
where $\bm\mu^{(G+1)}$ is the common mean shared by all $n$ observations.
Given a feature partition ${\mathcal A_1,\dots,\mathcal A_{G+1}}$, let $g_j$ denote the group label for feature $j\in[p]$, i.e., $j\in\mathcal A_{g_j}$. 
Let $\mathbf P^* = (\mathbf P_1^*,\dots,\mathbf P_G^*)$ collect the true clustering matrices associated with the $G$ informative feature groups. For each $g\in [G]$, the matrix $\mathbf P_g^*$ is unknown and needs to be estimated according to the group-specific sample clustering structure.
Note that, once the informative groups ${\mathcal A_1,\dots,\mathcal A_G}$ are identified, the non-informative group $\mathcal A_{G+1}$ is determined accordingly.

\subsection{Local spectral clustering}

To incorporate the sample clustering structure into feature clustering, we construct, for each feature $j\in[p]$, a similarity matrix that captures pairwise proximity among observations. In this paper, we adopt a Gaussian kernel to define the similarity matrices, which serve as the basis for the objective function, the computational algorithm, and the subsequent theoretical analysis. Specifically, for each feature $j\in[p]$, let $\mathbf K_j\in[0,1]^{n\times n}$ denote the corresponding similarity matrix, whose $(i,i')$-th entry is defined as $K_{j,ii'}=\exp\Big\{-\frac{(x_{i,j}-x_{i',j})^2}{2\theta_j^2}\Big\}$, where $\theta_j$ is a feature-specific scaling parameter for feature $j$ (The choice of parameters $\{\theta_j\}_{j=1}^p$ will be discussed in detail in Section 2.4).
Then, given the number of informative feature groups $G$, solving \eqref{obj} (defined later) yields $(G+1)$ clustering matrices $\{\widehat{\mathbf P}_1,\dots,\widehat{\mathbf P}_G,\mathbf E_n\}$ along with the corresponding feature partition $\{\widehat{\mathcal A}_1,\dots,\widehat{\mathcal A}_{G+1}\}$. Lastly, we apply the reduced-rank spectral clustering algorithm \citep{zhou2019analysis,amini2021concentration} to each $\widehat{\mathbf P}_g$, obtaining $G$ group-specific clustering structures. The overall procedure of local spectral clustering, used to construct $G$ distinct feature groups and their associated sample clustering structures, is summarized as follows:

\begin{enumerate}[\it Step 1.]
    \item Construct Gaussian kernel matrices $\{\mathbf K_j;j\in[p]\}$, where $K_{j,ii'}=\exp\left\{-\frac{(x_{i,j}-x_{i',j})^2}{2\theta_j^2}\right\}$.
    
    \item Solve \eqref{obj} to obtain $G$ center clustering matrix estimators $\{\widehat{\mathbf P}_g,g\in[G]\}$ and the associated feature group indices $\{\widehat{\mathcal A}_g,g\in[G]\}$. The non-informative feature group index $\widehat{\mathcal A}_{G+1}$ is then determined accordingly.
    
    \item For each $g\in[G]$, calculate the rank-$M_g$ truncated eigenvalue decomposition (EVD) of $\widehat{\mathbf P}_g$, 
    \[\widehat{\mathbf P}_{g}^{(M_g)}=\widehat{\mathbf U}_g^{(M_g)}\widehat{\bm\Lambda}_g^{(M_g)}\widehat{\mathbf U}_g^{(M_g)\top}.
    \]
    Specifically, let $\widehat{\mathbf P}_g=\widehat{\mathbf U}_g\widehat{\bm\Lambda}_g\widehat{\mathbf U}_g^{\top}$ be the full EVD of $\widehat{\mathbf P}_g$, where $\widehat{\bm\Lambda}_g=\text{diag}(\widehat\lambda_1,\dots,\widehat\lambda_n)$ with $|\widehat\lambda_1|\ge\cdots\ge|\widehat\lambda_n|$.
    Then $\widehat{\bm\Lambda}_g^{(M_g)}=\text{diag}(\widehat\lambda_1,\dots,\widehat\lambda_{M_g})$, and $\widehat{\mathbf U}_g^{(M_g)}\in\mathbb R^{n\times M_g}$ consists of the corresponding $M_g$ leading eigenvectors.
    
    \item For each $g\in[G]$, apply a constant-factor $K$-means algorithm to the rows of $\widehat{\mathbf U}_g^{(M_g)}\widehat{\bm\Lambda}_g^{(M_g)}$ to obtain the group-specific sample membership matrices $\{\widehat{\mathbf Z}_g:g\in[G]\}$.
\end{enumerate}

In Step 2, we consider the following objective function to simultaneously identify the informative feature groups and estimate their corresponding clustering matrices:
\begin{equation}\label{obj}
\begin{aligned}
& \left(\widehat{\mathbf P},\widehat{\mathbf g}\right):= \arg\max\left\{ \sum_{j=1}^p\left\langle\mathbf K_j-\tau\mathbf E_n,\mathbf P_{g_j}\right\rangle\right\}\\
\text{subject to}\quad&  0\le P_{g,ii'}\le 1 \quad \forall\ g\in[G],\  i,i'\in[n],\quad \text{and}\quad \mathbf P_{G+1}=\mathbf E_n,\\
\end{aligned}
\end{equation}
where $\mathbf P=(\mathbf P_1,\dots,\mathbf P_G)$, and $\tau\in(0,1)$ is a tuning parameter (The choice of tuning parameter $\tau$ will be discussed in detail in Section 2.4). 
Inspired by the classical Lloyd's algorithm \citep{macqueen1967some} for $K$-means clustering, we develop Algorithm 1 to solve \eqref{obj} in Step 2 of LSC. The procedure is described as follows.
\begin{center}\label{algo1}
    \textbf{Algorithm 1.}
\end{center}
\begin{enumerate}[(1)]
    \item Initialize a feature partition $\widehat{\mathcal A}^{(0)}=\{\widehat{\mathcal A}_1^{(0)},\dots,\widehat{\mathcal A}_{G+1}^{(0)}\}$ based on  clustering the $p\times n^2$ matrix $\mathbf K=(\mathrm{vec}(\mathbf K_1),\dots,\mathrm{vec}(\mathbf K_p))^\top$. Assign the group label $\widehat g_j^{(0)}=g$ if $j\in\widehat{\mathcal A}_g^{(0)}$ for $g\in[G+1]$.
    \item Given $\widehat{\mathcal A}^{(0)}$, for each $g\in[G+1]$, compute
    \[
    \widehat{\mathbf P}_g^{(0)}:=\ \arg\max_{\mathbf P_g}\quad \left\langle\left({|\widehat{\mathcal A}_g^{(0)}|^{-1}}\textstyle\sum_{j\in\widehat{\mathcal A}_g^{(0)}}\mathbf K_j\right)-\tau\mathbf E_n,\mathbf P_g\right\rangle.
    \]
    \item Given $\{\widehat{\mathbf P}_1^{(0)},\dots,\widehat{\mathbf P}_{G+1}^{(0)}\}$, compute
    \[
    \widetilde g:=\ \arg\min_{1\le g\le G+1}\quad \left\|\widehat{\mathbf P}_g^{(0)}-\mathbf E_n\right\|_F^2,
    \]
    and set $\widehat{\mathbf P}_{\widetilde g}^{(0)}=\mathbf E_n$. Without loss of generality, we can swap labels $\widetilde g$ and $G+1$, that is, let $\widehat g_j^{(0)}=\widetilde g$ if $j\in\widehat{\mathcal A}_{G+1}^{(0)}$ and $\widehat g_j^{(0)}=G+1$ if $j\in\widehat{\mathcal A}_{\widetilde g}^{(0)}$. This ensures that $\widehat{\mathbf P}_{G+1}^{(0)}=\mathbf E_n$.
    \item For iteration $w\ge 1$, repeat the following steps until the feature membership estimates do not change:
    \begin{enumerate}[(a)]
    \item Given $\{\widehat{\mathbf P}_1^{(w-1)},\dots,\widehat{\mathbf P}_G^{(w-1)},\widehat{\mathbf P}_{G+1}^{(w-1)}=\mathbf E_n\}$, update
    \[
    \widehat g_j^{(w)}:=\ \arg\max_{g\in[G+1]}\quad \left\langle\mathbf K_j-\tau\mathbf E_n,\widehat{\mathbf P}_g^{(w-1)}\right\rangle.
    \]
    \item Given $\{\widehat g_1^{(w)},\dots, \widehat g_p^{(w)}\}$, form the partition $\widehat{\mathcal A}^{(w)}=\{\widehat{\mathcal A}_1^{(w)},\dots,\widehat{\mathcal A}_{G+1}^{(w)}\}$. For each $g\in[G]$, compute
    \[
    \widehat{\mathbf P}_g^{(w)}:=\ \arg\max_{\mathbf P_g}\quad \left\langle\left({|\widehat{\mathcal A}_g^{(w)}|^{-1}}\textstyle\sum_{j\in\widehat{\mathcal A}_g^{(w)}}\mathbf K_j\right)-\tau\mathbf E_n,\mathbf P_g\right\rangle.
    \]
    \end{enumerate}
\end{enumerate}

\begin{remark}\label{obj_interpretation}
To gain insight into the objective function \eqref{obj}, consider the simplified case with $p=1$, under which \eqref{obj} reduces to
\begin{equation}\label{obj_example}
\widehat{\mathbf Q}_1:= \arg\max\left\langle\mathbf K_1-\tau\mathbf E_n,\mathbf Q_1\right\rangle.
\end{equation}
It can be readily verified that \eqref{obj_example} admits a closed-form solution given by
\[
\widehat{Q}_{1,ii'}=\left\{
\begin{aligned}
&1\quad\mathrm{if}\quad K_{1,ii'}-\tau>0, \\
&0\quad\mathrm{otherwise}. \\
\end{aligned}
\right.
\]
Under an appropriate choice of the scaling parameter $\theta_1$, the similarity values satisfy a clear separation:
for any $i,i'$ from the same cluster and any $i,i''$ from different clusters, $K_{1,ii'}$ is close to 1 while $K_{1,ii''}$ is close to 0, provided that the distance between any two cluster centers (i.e., the signal-to-noise ratio) is sufficiently large. 
In this case, there exists $0<\tau<1$ such that $K_{1,ii'}>\tau>K_{1,ii''}$, which ensures that $\widehat{\mathbf Q}_1$ exactly recovers the true clustering matrix. Extending this argument, given the true feature partition, one can guarantee that $\widehat{\mathbf P}_g=\mathbf P_{g}^*$ for all $g\in[G]$ as long as the average signal-to-noise ratio within each informative group is sufficiently large. Furthermore, if the matrices $\{\mathbf P_g^{*}\}_{g\in[G]}$ are well separated, the corresponding group-specific sample partitions can be consistently recovered with high probability, for example via a subsequent $K$-means procedure.
\end{remark}

\begin{remark}\label{rem_ini}
Algorithm 1 concludes convergence once the feature memberships at iterations $w$ and $(w+1)$ coincide. As an adaptation of Lloyd's algorithm, it is susceptible to convergence to local optima. 
To mitigate this issue, \cite{arthur2007k} proposed the $K$-means++ initialization, which selects good starting points and improves both convergence speed and clustering accuracy when followed by Lloyd’s iterations. 
Directly applying the $K$-means++ algorithm in Step (1), however, is computationally intensive due to the $n^2$ dimensionality of $\mathbf K$. To address this issue, inspired by \cite{chan2017efficient}, we combine $K$-means++ with random projection to obtain an efficient and reliable initialization. To further reduce sensitivity to initialization, we run the algorithm with 20 such initializations and select the solution that achieves the largest objective function value of \eqref{obj}.
\end{remark}

\begin{remark}\label{rem_approx}
In Step 4, a natural choice for the constant-factor $K$-means algorithm is the $(1+\omega)$-approximate $K$-means algorithm \citep{kumar2004simple}. Another common alternative is the $K$-means++ algorithm \citep{arthur2007k}, which admits a $(1+\log M_g)$-approximation guarantee. Owing to its theoretical guarantee and widespread use \citep{zhou2019analysis,liu2023clustering}, we adopt the $K$-means++ algorithm in our implementation.
\end{remark}

\subsection{Tuning parameter selection}
In this subsection, we provide a detailed discussion of the selection of the tuning parameters $\{\theta_j\}_{j=1}^p$, $\tau$, and $\{M_g\}_{g=1}^G$.
\begin{itemize}
\item \textbf{Choice of  $\theta_j$}.
Inspired by \cite{yan2021covariate}, we adopt the data-driven procedure of \cite{shi2009data} to select the scale parameter $\theta_j$. The key idea is to ensure that a sufficient proportion (say $\beta\times 100\%$) of pairwise distances falls within the effective range of the kernel function for most (say $\alpha\times 100\%$) data points. Under this principle, each $\theta_j$ is determined by
\[
\theta_j=\frac{\alpha\ \text{quantile of}\ \{q_{1j},\dots,q_{nj}\}}{\sqrt{\alpha\ \text{quantile of}\ \chi_1^2}},
\]
where $\chi_1^2$ denotes a Chi-squared distribution with degree of freedom 1, and $q_{ij}$ is defined as the $\beta$ quantile of $|x_{ij}-x_{lj}| (l\in[n])$ for feature $j$.

\item \textbf{Choice of  $\tau$}.
From the objective function \eqref{obj}, it is clear that $\tau$ acts as a threshold separating within-cluster and between-cluster observations. Motivated by \cite{srivastava2023robust}, the optimal $\tau$ is determined by
\[
\tau=\exp\left\{-\frac{t_{1-\alpha}}{2}\right\},
\]
where $t_{1-\alpha}$ is the $\alpha$ quantile of the $\chi_1^2$ distribution. This choice is obtained by setting the distance in the Gaussian kernel to the $\alpha$ quantile of $\{q_{1j},\dots,q_{nj}\}$ for each $j\in[p]$.

\item \textbf{Choice of  $M_g$}. 
In the proposed approach, we assume that the number of informative feature clusters $G$ is known in advance. In practice, such as in the exploratory clustering of biological omics data, the number of clustering structures is often determined based on relevant prior knowledge about the data.
After constructing the center matrix estimators $\widehat{\mathbf P}_g$, we need to select an appropriate $M_g$ to implement Step 3 in Algorithm 1 to obtain a nested sample partition within the $g$-th cluster. Specifically, we can use the eigen-gap heuristic \citep{von2007tutorial} based on the normalized Laplacian matrix $\mathbf L:=\mathbf I_n-\mathbf D_g^{-1/2}\widehat{\mathbf P}_g\mathbf D_g^{-1/2}$, where $\mathbf D_g=\text{diag}(\widehat{\mathbf P}_g\bm 1_n)$ is the $n$-dimensional diagonal matrix. It is easy to show $\mathbf L$ is positive semi-definite, and we can choose the optimal $\widehat M_g$ by
\[
\widehat M_g:=\ \arg\max_{r}\ \kappa_{r+1}(\mathbf L)-\kappa_{r}(\mathbf L),
\]
where $\kappa_r(\mathbf L)$ is the $r$-th smallest eigenvalue of $\mathbf L$.
\end{itemize}
In the subsequent simulations and real data analysis, we set $\alpha=0.9$ and $\beta=0.2$, which lead to stable and favorable numerical performance.

\section{Theory}\label{sec3}
In this section, we first establish the non-asymptotic upper bounds of the estimation error for the group-oracle estimators (defined later) and discuss different signal-to-noise conditions for the recovery of the sample cluster structures. Afterwards, we show the asymptotic recovery of feature group structure by verifying that the group-oracle estimator is exactly the local maximizer of \eqref{obj} with high probability. Finally, we turn to the group-specific sample clustering structures and further establish the uniform non-asymptotic upper bound of the misclassification error for the LSC estimator.

To define the recovery of the clustering structure for samples, we first introduce some definitions and notations from \cite{sun2021convex} and \cite{yagishita2024pursuit}.

\begin{defin}[Recovery of cluster structure]\label{def_structure}
    For any $g\in[G]$, let $\overline{\mathcal G}_g^*:=\{\overline{\mathcal G}_{g,1}^*,\dots,\overline{\mathcal G}_{g,L_g}^*\}$ with $L_g\le M_g$ be another partition of $[n]$. We call $\overline{\mathcal G}_g^*$ a coarsening of $\mathcal G_g^*$ if for any $\overline{\mathcal G}_{g,l}^*\in\overline{\mathcal G}_g^*$ there exists $\mathcal T_{l}\subset[M_g]$ such that $\overline{\mathcal G}_{g,l}^*=\cup_{m_g\in\mathcal T_{l}}\mathcal G_{g,m_g}^*$.
    \begin{enumerate}[(a)]
        \item When $\overline{\mathcal G}_g^*=\mathcal G_g^*$, we say that $\overline{\mathcal G}_g^*$ perfectly recovers $\mathcal G_g^*$.
        \item  Moreover, $\overline{\mathcal G}_g^*$ is called the trivial coarsening if $\overline{\mathcal G}_g^*=\{[n]\}$. Otherwise, it is called a strict non-trivial coarsening. 
     
    \end{enumerate}
\end{defin}
For each $j\in\mathcal A_g^*$ and $m,m'\in[M_g]$, the distance between any pair of clusters ${\mathcal G}_{g,m}^{*}$ and ${\mathcal G}_{g,m'}^{*}$ is defined by $d_{mm',j}^{(g)}=\big|\mu_{m,j}^{(g)}-\mu_{m',j}^{(g)}\big|$. For each $j\in\mathcal A_g^*$ and $l\in[L_g]$, we denote $d_{l,j}^{(g)}=\max_{m\ne m'\in\mathcal T_l}d_{mm',j}^{(g)}$ with $\mathcal T_l\subset [M_g]$.
Besides, for each $j\in\mathcal A_g^*$ and $l\ne l'\in[L_g]$, we denote $d_{ll',j}^{(g)}=\min_{m\in\mathcal T_l,m'\in\mathcal T_{l'}}d_{mm',j}^{(g)}$ with $\mathcal T_l,\mathcal T_{l'}\subset [M_g]$.
Accordingly, for each $j\in\mathcal A_g^*$, let $d_{\min,j}^{(g)}=\min_{l\ne l'\in[L_g]}d_{ll',j}^{(g)}$ and $d_{\max,j}^{(g)}=\max_{l\in[L_g]}d_{l,j}^{(g)}$. Then, the minimal signal-to-noise ratio (SNR) for $j\in\mathcal A_g^*$ is defined by $\mathrm{SNR}_{\min,j}^{(g)}:=d_{\min,j}^{(g)}/\sigma$.

\begin{defin}[Oracle group estimator]\label{def_oracle}
    Given the true group partition $\{\mathcal A_1^*,\dots,\mathcal A_{G+1}^*\}$, the oracle group estimator $\widehat{\mathbf P}^{\mathrm{or}}=(\widehat{\mathbf P}^{\mathrm{or}}_1,\dots,\widehat{\mathbf P}^{\mathrm{or}}_G)$ is defined as
    \[
    \widehat{\mathbf P}^{\mathrm{or}}:=\arg\max_{\mathbf P\in[0,1]^{n\times nG}}\sum_{g=1}^G\left\{ \sum_{j\in\mathcal A_g^*}\left\langle\mathbf K_j-\tau\mathbf E_n,\mathbf P_{g}\right\rangle\right\}.
    \]
    More specifically, let ${\mathbf K}_g^{\mathrm{or}}=|{\mathcal A}_g^*|^{-1}\sum_{j\in{\mathcal A}_g^{*}}{\mathbf K}_j$ for $g\in[G]$. Then, the oracle group estimator $\widehat{\mathbf P}^{\mathrm{or}}$ can be defined equivalently by
\begin{equation}\label{ora_obj}
\widehat{\mathbf P}_g^{\mathrm{or}}:=\arg\max_{\mathbf P_g\in[0,1]^{n\times n}} \left\langle\mathbf K_g^{\mathrm{or}}-\tau\mathbf E_n,\mathbf P_{g}\right\rangle\quad\text{for}\quad g\in[G].
\end{equation}
\end{defin}

\subsection{Error analysis of the oracle estimator}
In our theoretical analysis, we let $M_{\max}=\max_{g\in[G]}M_g$ and assume that both $G$ and $M_{\max}$ are finite, and $\epsilon_{ij}$ are independent across observations ($i$) and features ($j$).

\begin{condition}\label{con_snr1}
For the $g$-th ($g\in[G]$) group, we let $d_g=s_g^{-1}\sum_{j\in\mathcal A_g^*}d_{\min,j}^{(g)}$ and $\theta_j=\theta_{j'}=\vartheta_g:=\kappa_{2,g}d_g$ for any $j,j'\in\mathcal A_g^*$, where $\kappa_{2,g}>0$ is a constant. 
Given a constant $0<\kappa_{1,g}<1$, let 
\[
\begin{aligned}
&\xi_g^{\mathrm{in}}=\frac{1-\kappa_{1,g}}{s_g}\times\sum_{j\in\mathcal A_g^*}\exp\left\{-\frac{(d_{\max,j}^{(g)})^2+2\sigma^2}{2\vartheta_g^2}\right\}\qquad\text{and}\\
&\xi_{g}^{\mathrm{out}}=\frac{1+\kappa_{1,g}}{s_g}\times\sum_{j\in\mathcal A_g^*}\left(
\exp\left\{
-\frac{(d_{\min,j}^{(g)})^2}{8\vartheta_g^2}
\right\}+2\exp\left\{-\frac{(d_{\min,j}^{(g)})^2}{16\sigma^2}\right\}
\right).
\end{aligned}
\]
Assume that
\begin{equation}\label{snr_general}
\xi_g^{\mathrm{in}}>\xi_g^{\mathrm{out}}\quad \text{under}\quad 1<L_g\le M_g.
\end{equation}
Note that when $L_g=M_g$, $d_{\max,j}^{(g)}=0$ for any $j\in\mathcal A_g^*$.
\end{condition}
Condition \ref{con_snr1} works as a necessary condition to recover $\overline{\mathcal G}_g^{*}$, a non-trivial coarsening of $\mathcal G_g^{*}$, with high probability. Moreover, it is noteworthy that to ensure the validity of \eqref{snr_general}, it is not necessary for the signal of each feature in the $g$-th group to be sufficiently strong; rather, it suffices that the average signal across the $g$-th group is strong enough.

\begin{remark}\label{rem_snr_condition}
Condition \ref{con_snr1} holds if the SNR is sufficiently large. To see this, consider a special case in which $d_{\min,j}^{(g)}=d_{\min,j'}^{(g)}$ and $d_{\max,j}^{(g)}=d_{\max,j'}^{(g)}$ hold for any $j,j'\in\mathcal A_g^*$. Then, \eqref{snr_general} is reduced to
\begin{equation}\label{snr_reduced}
    \exp\left\{-\frac{(d_{\max,j}^{(g)}/\sigma)^2+2}{2(\vartheta_g/\sigma)^2}\right\}> \frac{1+\kappa_{1,g}}{1-\kappa_{1,g}}\cdot\left(\exp\left\{-\frac{(\mathrm{SNR}_{\min,j}^{(g)})^2}{8(\vartheta_g/\sigma)^2}\right\}+2\exp\left\{-\frac{(\mathrm{SNR}_{\min,j}^{(g)})^2}{16}\right\}\right).
\end{equation}
If we further assume that $\mathrm{SNR}_{\min,j}^{(g)}=c$ holds for a constant $c>0$, \eqref{snr_reduced} can be rewritten as
\begin{equation}\label{snr_reduced1}
    \exp\left\{-\frac{(d_{\max,j}^{(g)}/\sigma)^2+2}{2(\kappa_{2,g}c)^2}\right\}> \frac{1+\kappa_{1,g}}{1-\kappa_{1,g}}\cdot\left(\exp\left\{-\frac{1}{8\kappa_{2,g}^2}\right\}+2\exp\left\{-\frac{c^2}{16}\right\}\right).
\end{equation}
In this case, \eqref{snr_reduced1} can be satisfied if there exists a sufficiently small constant $\kappa_{2,g}$ and a sufficiently large constant $c$, such that $c\kappa_{2,g}$ is sufficiently large.

\end{remark}

\begin{theorem}\label{thm1} 
Assume that Condition \ref{con_snr1} holds for any $g\in[G]$. 
If there exists a constant $0<\gamma<1$ such that $\max_{g\in[G]}\xi_g^{\mathrm{out}}<\gamma<\min_{g\in[G]}\xi_g^{\mathrm{in}}$, after taking $\tau=\gamma$, we have that with probability at least $1-2GM_{\max}/n_{\min}$,
\[
\sup_{g\in[G]}\left\|\widehat{\mathbf P}_g^{\mathrm{or}}-\overline{\mathbf P}_g^{*}\right\|_1\le r_n:=Cn^2\cdot\max\left\{\exp\left[-\kappa s_{\min}\right],\frac{\log n_{\min}}{n_{\min}}\right\},
\]
where $s_{\min}=\min_{g\in[G]}s_g$, $n_{\min}=\min_{g\in[G],m_g\in[M_g]}n_{m_g}^{(g)}>2GM_{\max}$, and $n_{\min}\asymp n$.
The constant $\kappa>0$ depends on $\{\kappa_{1,g},\kappa_{2,g}\}_{g=1}^G$, and $C>0$ is also a constant. Here $\overline{\mathbf P}_g^{*}$ is a clustering matrix corresponding to $\overline{\mathcal G}_g^{*}$ (a strict non-trivial coarsening of $\mathcal G_g^{*}$).
\end{theorem}

Under Condition \ref{con_snr1}, if \eqref{snr_general} holds under $L_g=M_g$ for all $g\in[G]$ and $\max_{g\in[G]}\xi_g^{\mathrm{out}}<\min_{g\in[G]}\xi_g^{\mathrm{in}}$, Theorem \ref{thm1} tells us $\mathcal G_1^{*},\dots,\mathcal G_G^{*}$ can be perfectly recovered with high probability if the true feature group structure is known. 
However, under $L_g=M_g$ for all $g\in[G]$, \eqref{snr_general} does not always hold in practice. For example, if the centers of certain clusters are relatively close to each other but distant from the centers of the other clusters, \eqref{snr_general} may not hold under $L_g=M_g$ for all $g\in[G]$. But \eqref{snr_general} could be satisfied when $1<L_g<M_g$ for some $g\in\mathcal K_1\subset [G]$, and for those $g\in[G]\setminus \mathcal K_1$, \eqref{snr_general} can hold under $L_g = M_g$.
Moreover, it may not be possible to achieve $\max_{g\in[G]}\xi_g^{\mathrm{out}}<\min_{g\in[G]}\xi_g^{\mathrm{in}}$ when \eqref{snr_general} holds under $L_g=M_g$ for all $g\in[G]$. However, it could be satisfied by allowing \eqref{snr_general} to hold under $L_g<M_g$ for some $g\in\mathcal K_2\subset [G]$. Simultaneously, for those $g\in[G]\setminus \mathcal K_2$, \eqref{snr_general} can hold under $L_g = M_g$.
In such scenarios, Theorem \ref{thm1} tells us that, for $g\in\mathcal K_1$ (or $\mathcal K_2$), a strict non-trivial coarsening of $\mathcal G_g^{*}$, i.e., $\overline{\mathcal G}_g^{*}$, can be recovered without giving rise to a chaotic clustering structure with high probability for sufficiently large $n$.

\subsection{Asymptotic recovery of group structure}

For theoretical convenience, in this subsection and the subsequent subsections, we assume that Condition \ref{con_snr1} holds under $L_g=M_g$ and $\max_{g\in[G]}\xi_g^{\mathrm{out}}<\gamma<\min_{g\in[G]}\xi_g^{\mathrm{in}}$. Therefore, we can clearly show that the true feature group structure can be asymptotically recovered under true sample clustering structures. This result can be similarly adjusted for asymptotic recovery under non-trivial coarsening structures, with corresponding changes to the conditions in Theorem \ref{thm1}.
Before showing that the oracle estimator $\widehat{\mathbf P}^{\mathrm{or}}$ is a strictly local maximizer of \eqref{obj} with high probability, we first provide Condition \ref{con_snr2} and then quantify the differences between $\mathbf P_g^*$ and $\mathbf P_{g'}^*$. 

\begin{condition}\label{con_snr2}
For the $g$-th ($g\in[G+1]$) group and the $j$-th ($j\in\mathcal A_g^*$) feature,
Let $\theta_j=\kappa_jd_{\min,j}^{(g)}$ for any $j\in\mathcal A_g^*$, where $\kappa_j>0$ is a constant.
Given a constant $0<\kappa_j'<1$, let 
\[
\begin{aligned}
&\zeta_j^{\mathrm{in}}=(1-\kappa_j')\exp\left\{-\frac{\sigma^2}{\theta_j^2}\right\}
\qquad\text{and}\\
&\zeta_{j}^{\mathrm{out}}=(1+\kappa_j')\left(
\exp\left\{
-\frac{(d_{\min,j}^{(g)})^2}{8\theta_j^2}
\right\}+2\exp\left\{-\frac{(d_{\min,j}^{(g)})^2}{16\sigma^2}\right\}
\right).
\end{aligned}
\]
Assume that
\[
\zeta_j^{\text{in}}>\zeta_j^{\text{out}}\quad {for}\quad j\in\mathcal A_g^* \quad and\quad g\in[G].
\]
\end{condition}
Similar to Remark \ref{rem_snr_condition}, it can be verified that $\zeta_j^{\text{in}}>\zeta_j^{\text{out}}$ holds for sufficiently small $\kappa_j'$ and sufficiently large $\text{SNR}_{\min,j}^{(g)}$, such that $\kappa_j'\text{SNR}_{\min,j}^{(g)}$ is large enough.

Now, we quantify the differences between $\mathbf P_g^*$ and $\mathbf P_{g'}^*$. In particular, given a pair of $(g,g')$ with $g\ne g'\in[G+1]$, by suitably permuting the matrix $\mathbf P_g^*-\mathbf P_{g'}^*$, we can always find a partition of the observations $\mathcal H^{(gg')}=\{\mathcal H_1^{(gg')},\dots,\mathcal H_{h_{gg'}}^{(gg')}\}$, under which for any $h,h'\in[h_{gg'}]$,
\[
D_{hh'}^{gg'}:=\left\{\sum_{i\in\mathcal H_h^{(gg')},i'\in\mathcal H_{h'}^{(gg')}}(P_{g,ii'}^*-P_{g',ii'}^*)\right\}\in\left\{\left|\mathcal H_{h}^{(gg')}\right|\cdot\left|\mathcal H_{h'}^{(gg')}\right|,-\left|\mathcal H_{h}^{(gg')}\right|\cdot\left|\mathcal H_{h'}^{(gg')}\right|,0\right\}.
\]

\begin{remark}\label{rem11}
    The partition $\mathcal H^{(gg')}=\{\mathcal H_1^{(gg')},\dots,\mathcal H_{h_{gg'}}^{(gg')}\}$ is constructed for theoretical analysis. As shown in Figure 8 (supplementary material), without rearranging matrix $\mathbf P_g^*-\mathbf P_{g'}^*$, the entries would appear as scattered points, as shown in Figure 8(3), where each point represents either 1 or -1. Such a dispersed pattern makes it difficult to apply large-sample theory for quantitative analysis. In contrast, after rearrangement, the pattern shown in Figure 8(6) exhibits a clear block structure, which facilitates the quantitative comparison of the differences between $\mathbf P_g^*$ and $\mathbf P_{g'}^*$. 
\end{remark}

Note that $h_{gg'}\le M_gM_{g'}$ is the number of groups in $\mathcal H^{(gg')}$ related to the differences between $\mathbf P_g^*$ and $\mathbf P_{g'}^*$. Thus, for each pair of $(g,g')$ with $g\ne g'\in[G+1]$, we only need to focus on two corresponding subsets of $\{(h,h'):1\le h<h'\le h_{gg'}\}$, defined by
\[
\begin{aligned}
&\mathcal D_1^{gg'}=\left\{(h,h'):D_{hh'}^{gg'}=\left|\mathcal H_{h}^{(gg')}\right|\cdot\left|\mathcal H_{h'}^{(gg')}\right|\right\},\\
&\mathcal D_{-1}^{gg'}=\left\{(h,h'):D_{hh'}^{gg'}=-\left|\mathcal H_{h}^{(gg')}\right|\cdot\left|\mathcal H_{h'}^{(gg')}\right|\right\}.
\end{aligned}
\]
Let $|\mathcal D_{1}|:=\max_{g'\ne g\in[G+1]}|\mathcal D_{1}^{gg'}|$, $|\mathcal D_{-1}|:=\max_{g'\ne g\in[G+1]}|\mathcal D_{-1}^{gg'}|$ and $n_h^{(gg')}=\left|\mathcal H_{h}^{(gg')}\right|$. 
For some $g\ne g'\in[G+1]$, let $\mathcal D_{\mathrm{I}}^{(gg')}$ be a subset of $\mathcal D_1^{gg'}\cup\mathcal D_{-1}^{gg'}$ such that, for any $(h,h')\in\mathcal D_{\mathrm{I}}^{(gg')}$, $\min\left\{n_h^{(gg')},n_{h'}^{(gg')}\right\}\asymp n$, and $\mathcal D_{\mathrm{II}}^{(gg')}$ be another subset of $\mathcal D_1^{gg'}\cup\mathcal D_{-1}^{gg'}$ such that, for any $(h,h')\in\mathcal D_{\mathrm{II}}^{(gg')}$, $\min\left\{n_h^{(gg')},n_{h'}^{(gg')}\right\}\ll n$. 
Let 
\[
n_{\min}^{\mathrm{diff}}:=\min_{(h,h')\in\mathcal D_{\mathrm{I}}^{(gg')}}\min\left\{n_h^{(gg')},n_{h'}^{(gg')}\right\}.
\]

For any $g\in[G+1]$ and $(h,h')\in\mathcal D_1^{gg'}\cup\mathcal D_{-1}^{gg'}$, let $\mathbf P_{g,hh'}^*$ be the submatrix of $\mathbf P_g^*$ with rows and columns corresponding to $\mathcal H_{h}^{(gg')}$ and $\mathcal H_{h'}^{(gg')}$, respectively.
Define $\zeta_{\min}^{\text{in}}:=\min_{j\in[p]}\zeta_j^{\mathrm{in}}$ and $\zeta_{\max}^{\text{out}}:=\max_{j\in[p]}\zeta_j^{\mathrm{out}}$.
The separation between the two clustering matrices $\mathbf P_g^*$ and $\mathbf P_{g'}^*$ is defined by 
\[
\begin{aligned}
\Lambda_{gg'}:=&\ \min\{\zeta_{\min}^{\mathrm{in}}-\tau,\tau-\zeta_{\max}^{\mathrm{out}}\}\sum_{(h,h')\in\mathcal D_{\mathrm{I}}^{(gg')}}\left\|\mathbf P_{g,hh'}^*
-\mathbf P_{g',hh'}^*\right\|_1\\
&\qquad\qquad\quad -\max\{\tau,1-\tau\}\sum_{(h,h')\in\mathcal D_{\mathrm{II}}^{(gg')}}\left\|\mathbf P_{g,hh'}^*-\mathbf P_{g',hh'}^*\right\|_1.
\end{aligned}
\]

\begin{theorem}\label{thm2}  Suppose the conditions in Theorem \ref{thm1} hold.
Assume that Condition \ref{con_snr2} holds for any $g\in[G+1]$.
If there exists a $0<\gamma<1$ such that $\zeta_{\min}^{\mathrm{in}}>\gamma>\zeta_{\max}^{\mathrm{out}}$. Assume that $\min_{g\ne g'\in[G+1]}\Lambda_{gg'}> 2r_n$ ($r_n$ is defined in Theorem \ref{thm1}), then after taking $\tau=\gamma$, we have
    \[
    \mathbb P\left[\bigcap_{g=1}^G\left\{\widehat{\mathbf P}_g^{\mathrm{or}}=\widehat{\mathbf P}_g\right\}\right]\ge 1-2pG(|\mathcal D_1|+|\mathcal D_{-1}|)\cdot\exp\{-C'n_{\min}^{\mathrm{diff}}\},
    \]
    where $C'>0$ is a constant.
\end{theorem}
\begin{remark}\label{rem1}
    Note that $G(|\mathcal D_1|+|\mathcal D_{-1}|)< GM_{\max}^4<\infty$, if $\log p\ll n$, the result in Theorem \ref{thm2} holds with probability approaching 1. 
    Note that, $n_{\min}^{\mathrm{diff}}\asymp n$ leads to $\Lambda_{gg'}\asymp n^2-o(n^2)$. Besides, $r_n\ll n^2$ as $s_{\min}\rightarrow\infty$ and $n\rightarrow\infty$. Therefore, $\min_{g\ne g'\in[G+1]}\Lambda_{gg'}> 2r_n$ can be satisfied with high probability when $n$ and $s_{\min}$ are sufficiently large.
\end{remark}
Combining the results in Theorems \ref{thm1} and \ref{thm2}, the estimator $\widehat{\mathbf P}$ from \eqref{obj} is asymptotically equivalent to the oracle estimator $\widehat{\mathbf P}^{\mathrm{or}}$ from \eqref{ora_obj}. In addition, the oracle estimator $\widehat{\mathbf P}^{\mathrm{or}}$ is calculated based on the true feature group partition. Hence $\widehat{\mathbf P}^{\mathrm{or}}$ can recover the true feature group partition with high probability.

\subsection{Misclassification error}

Before showing the results about clustering consistency, we firstly introduce some definitions. Let $\mathbb H^{n\times M}$ be the set of hard (cluster) labels: $\{0,1\}$-valued $n\times M$ matrices where each row has exactly a single 1. For any two membership matrices $\mathbf Z,\mathbf Z'\in\mathbb H^{n\times M}$ with $\mathbf Z=(\mathbf z_1,\dots,\mathbf z_n)^\top$ and $\mathbf Z'=(\mathbf z_1',\dots,\mathbf z_n')^\top$, the average misclassification rate between them is denoted by $\overline{\mathrm{Mis}}(\mathbf Z,\mathbf Z')$, defined by
\[
\overline{\mathrm{Mis}}(\mathbf Z,\mathbf Z'):=\min_{\mathbf Q\in\mathcal Q}\ \frac{1}{n}\sum_{i=1}^n\mathbbm 1\{\mathbf z_i\ne \mathbf Q\mathbf z_i'\},
\]
where $\mathcal Q$ collects all $M\times M$ permutation matrices $\mathbf Q$.

\begin{theorem}\label{thm3}  Let $\widehat{\mathbf Z}_g$ be the estimated membership matrix obtained by applying spectral clustering on $\widehat{\mathbf P}_g^{(M_g)}$ based on a $(1+\omega)$-approximate K-means algorithm.
Suppose the conditions in Theorems \ref{thm1} and \ref{thm2} hold.
Then for sufficiently large $n$ and $s_{\min}$, with probability at least $1-2GM_{\max}/n_{\min}-2pG(|\mathcal D_1|+|\mathcal D_{-1}|)\cdot\exp\{-C'n_{\min}^{\mathrm{diff}}\}$, we have
\[
\sup_{g\in[G]}\overline{\mathrm{Mis}}\left(\mathbf P_g^*;\mathcal P_{\omega}(\widehat{\mathbf P}_g^{(M_g)})\right)\le 16M_{\max}(1+\omega)^2C''\cdot\max\left\{\exp[-\kappa s_{\min}],\frac{\log n_{\min}}{n_{\min}}\right\},
\]
where $C''>0$ is a constant, $\mathcal P_{\omega}(\widehat{\mathbf X})=\{\widetilde{\mathbf X}\in\mathbb M_{n,n'}^M:\|\widehat{\mathbf X}-\widetilde{\mathbf X}\|_F\le\omega\|\mathbf X-\widehat{\mathbf X}\|_F,\ \forall\mathbf X\in\mathbb M_{n,n'}^M\}$, and $\mathbb M_{n,n'}^M$ is a class of $M$-means matrices (defined in Appendix). Here, if $\mathbf X_1,\mathbf X_2\in\mathbb M_{n,n'}^M$ with the corresponding membership matrices $\mathbf Z_1,\mathbf Z_2\in\mathbb H^{n\times M}$, respectively, we define $\overline{\mathrm{Mis}}(\mathbf X_1,\mathbf X_2):=\overline{\mathrm{Mis}}(\mathbf Z_1,\mathbf Z_2)$.
\end{theorem}

\begin{remark}\label{rem2}
    Note that $n_{\min}\asymp n$ and $n_{\min}^{\mathrm{diff}}\asymp n$, thus as $n\rightarrow\infty$ and $s_{\min}\rightarrow\infty$, the misclassification errors for $G$ $M_g$-truncated clustering matrices tend to 0 with probability approaching 1.
\end{remark}
Combining the results in Theorems \ref{thm1} -- \ref{thm3}, the true feature group partition and the corresponding true (or coarsening) clustering structures can be perfectly recovered with high probability.

\section{Simulation studies}

In this section, we conduct numerical simulations to evaluate the finite-sample performance of the proposed LSC method. For comparison, we also consider five alternative approaches. (a) The multistep Local Spectral Clustering method (mLSC), which first applies the COSCI approach \citep{banerjee2017feature} to screen features with significant clustering signals, and then estimates a clustering matrix for each selected feature individually. The retained features are subsequently partitioned into a pre-specified number of $G$ groups, each associated with a center clustering matrix. Finally, applying the aforementioned spectral clustering procedure to each clustering matrix then yields the group-specific sample clustering structures.
(b) The Oracle method, which maximizes \eqref{ora_obj} to obtain the oracle group estimator.
(c) The Nonparametric Bayesian Local Clustering method (NBLC) proposed by \cite{lee2013nonparametric}, implemented via MCMC with 10000 iterations, of which the first 2000 iterations are discarded as burn-in.
(d) The Sparse $K$-means method (SKM) proposed by \cite{witten2010framework}, implemented using R package {\it sparcl}.
(e) The spectral clustering with feature selection method (SCFS) proposed by \cite{liu2023clustering}. 
Among these, method (a) serves as a multistep counterpart to our LSC method, while method (b) represents an infeasible oracle method. Together, they serve as two benchmark methods for assessing the advantages of the proposed one-step design and for validating our theoretical results. 
Since methods (d) and (e) are designed for sample clustering under a single homogeneous feature group, and assume that the number of clusters $M$ is known a priori, we denote their implementations with a specified number of clusters ($m$) by $\text{SKM(M=m)}$ and $\text{SCFS(M=m)}$, respectively.

\subsection{Simulation setup}

We consider $G=3$ informative feature groups, allowing for heterogeneous SNRs within each group, along with one non-informative feature group. For the first informative group, let $\mathbf X^{(1)}$ be an $n\times s$ informative data matrix consisting of two equally sized clusters $\mathcal G_{1,1}^*$ and $\mathcal G_{1,2}^*$. The cluster centers are given by $\bm\mu_{m}^{(1)}=(\mu_{m,1}^{(1)},\dots,\mu_{m,s}^{(1)})$, where $\mu_{1,j}^{(1)}=(2V_j^{(1)}-1)\times U_{1j}^{(1)}$, and $\mu_{2,j}^{(1)}=-(2V_j^{(1)}-1)\times U_{2j}^{(1)}$. Here, $\{U_{mj}^{(1)}\}_{m=1,j=1}^{2,s}$ are independent uniform variables with $\mathcal U(0.75,1.25)$, and $\{V_j^{(1)}\}_{j=1}^{s}$ are independent Bernoulli variables with mean 0.5. 
For each $i\in\mathcal G_{1,m}^*$, we generate $\mathbf X_{i\cdot}^{(1)}$, the transposed $i$th row of $\mathbf X^{(1)}$, from the multivariate normal distribution $\mathcal N(\bm\mu_m^{(1)},\mathbf I_{s})$, for $m\in[2]$. 
For the second informative group, let $\mathbf X^{(2)}$ be an $n\times s$ informative data matrix consisting of three clusters $\mathcal G_{2,1}^*$, $\mathcal G_{2,2}^*$, and $\mathcal G_{2,3}^*$ with cluster sizes 3:3:4. 
For the $m$-th cluster, its center is denoted by $\bm\mu_{m}^{(2)}=(\mu_{m,1}^{(2)},\dots,\mu_{m,s}^{(2)})$ with $\mu_{1,j}^{(2)}=(2V_j^{(2)}-1)\times U_{1j}^{(2)}$, $\mu_{2,j}^{(2)}=U_{2j}^{(2)}$, and $\mu_{3,j}^{(2)}=-(2V_j^{(2)}-1)\times U_{3j}^{(2)}$, where $\{U_{1j}^{(2)}\}_{j\in[7s/8]}$, $\{U_{3j}^{(2)}\}_{j\in[3s/4]}$, and $\{U_{3j}^{(2)}\}_{j\in[s]\setminus[7s/8]}$ are independent uniform variables with $\mathcal U(2,2.5)$, $\{U_{1j}^{(2)}\}_{j\in[s]\setminus[7s/8]}$, and $\{U_{3j}^{(2)}\}_{j\in[7s/8]\setminus[3s/4]}$ are independent uniform variables with $\mathcal U(1,1.5)$, and $\{U_{2j}^{(2)}\}_{j\in[s]}$ are independent uniform variables with $\mathcal U(-0.25,0.25)$, and $\{V_j^{(2)}\}_{j=1}^{s}$ are independent Bernoulli variables with mean 0.5.  For $i\in\mathcal G_{2,m}^*$, we generate $\mathbf X_{i\cdot}^{(2)}$ from the multivariate normal distribution $\mathcal N(\bm\mu_m^{(2)},\mathbf I_{s})$, for $m\in[3]$.
Similarly, let $\mathbf X^{(3)}$ be an $n\times s$ informative data matrix with four equally sized clusters $\mathcal G_{3,1}^*$, $\mathcal G_{3,2}^*$, $\mathcal G_{3,3}^*$, and $\mathcal G_{3,4}^*$. For the $m$-th cluster, its center is denoted by $\bm\mu_{m}^{(3)}=(\mu_{m,1}^{(3)},\dots,\mu_{m,s}^{(3)})$ with $\mu_{1,j}^{(3)}=(2V_j^{(3)}-1)\times U_{1j}^{(3)}$, $\mu_{2,j}^{(3)}=(2V_j^{(3)}-1)\times U_{2j}^{(3)}$, $\mu_{3,j}^{(3)}=-(2V_j^{(3)}-1)\times U_{3j}^{(3)}$, and $\mu_{4,j}^{(3)}=-(2V_j^{(3)}-1)\times U_{4j}^{(3)}$, where $\{U_{1j}^{(3)}\}_{j\in[7s/8]}$, $\{U_{4j}^{(3)}\}_{j\in[3s/4]}$, and $\{U_{4j}^{(3)}\}_{j\in[s]\setminus[7s/8]}$ are independent uniform variables with $\mathcal U(4,4.5)$, $\{U_{1j}^{(3)}\}_{j\in[s]\setminus[7s/8]}$, and $\{U_{4j}^{(3)}\}_{j\in[7s/8]\setminus[3s/4]}$ are independent uniform variables with $\mathcal U(2,2.5)$, and $\{U_{2j}^{(3)}\}_{j\in[s]}$ and $\{U_{3j}^{(3)}\}_{j\in[s]}$ are independent uniform variables with $\mathcal U(1,1.5)$, and $\{V_j^{(3)}\}_{j=1}^{s}$ are independent Bernoulli variables with mean 0.5. For $i\in\mathcal G_{3,m}^*$, we generate $\mathbf X_{i\cdot}^{(3)}$ from the multivariate normal distribution $\mathcal N(\bm\mu_m^{(3)},\mathbf I_{s})$ for $m\in[4]$.
Finally, let $\mathbf X^{(4)}$ be an $n\times (p-3s)$ non-informative feature group, whose entries are independently drawn from standard normal distribution $\mathcal N(0, 1)$. The full data matrix is then constructed as $\mathbf X=(\mathbf X^{(1)},\dots,\mathbf X^{(4)})$. We set $p=500$, $s\in\{16,24,32\}$, and $n\in\{100,200,400\}$.
Figure \ref{heatmap_G3} depicts the heatmaps for the four feature groups. Panels (a) -- (c) correspond to the three informative groups with distinct clustering structures, while panel (d) displays a representative subgroup of 24 features from the non-informative group.

\begin{figure}[H]
\centering
\subfloat[$\mathcal G_1^*$]{
 \includegraphics[height=4cm, width =4cm]{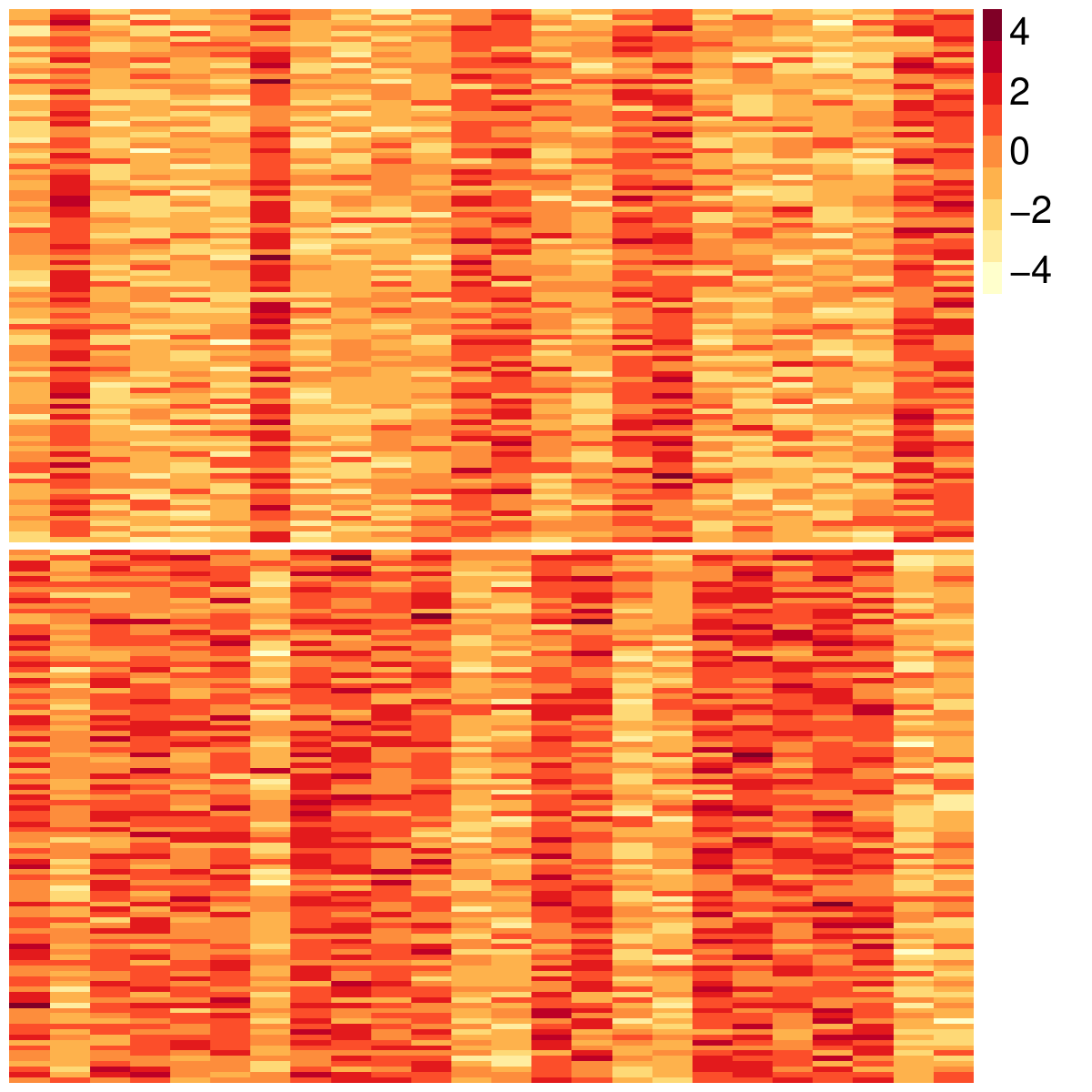}}
  \subfloat[$\mathcal G_2^*$]{
 \includegraphics[height=4cm, width =4cm]{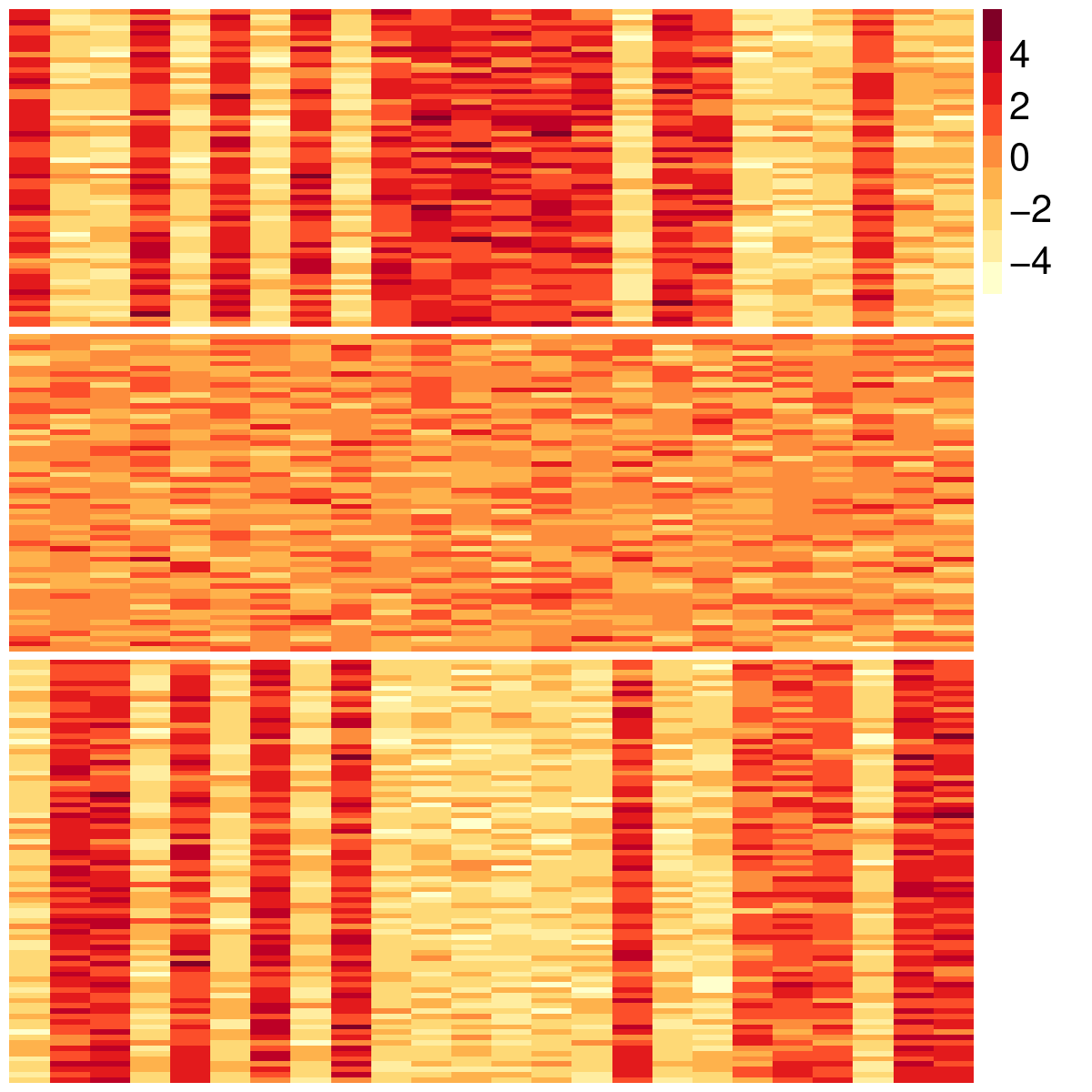}}
 \subfloat[$\mathcal G_3^*$]{
 \includegraphics[height=4cm, width =4cm]{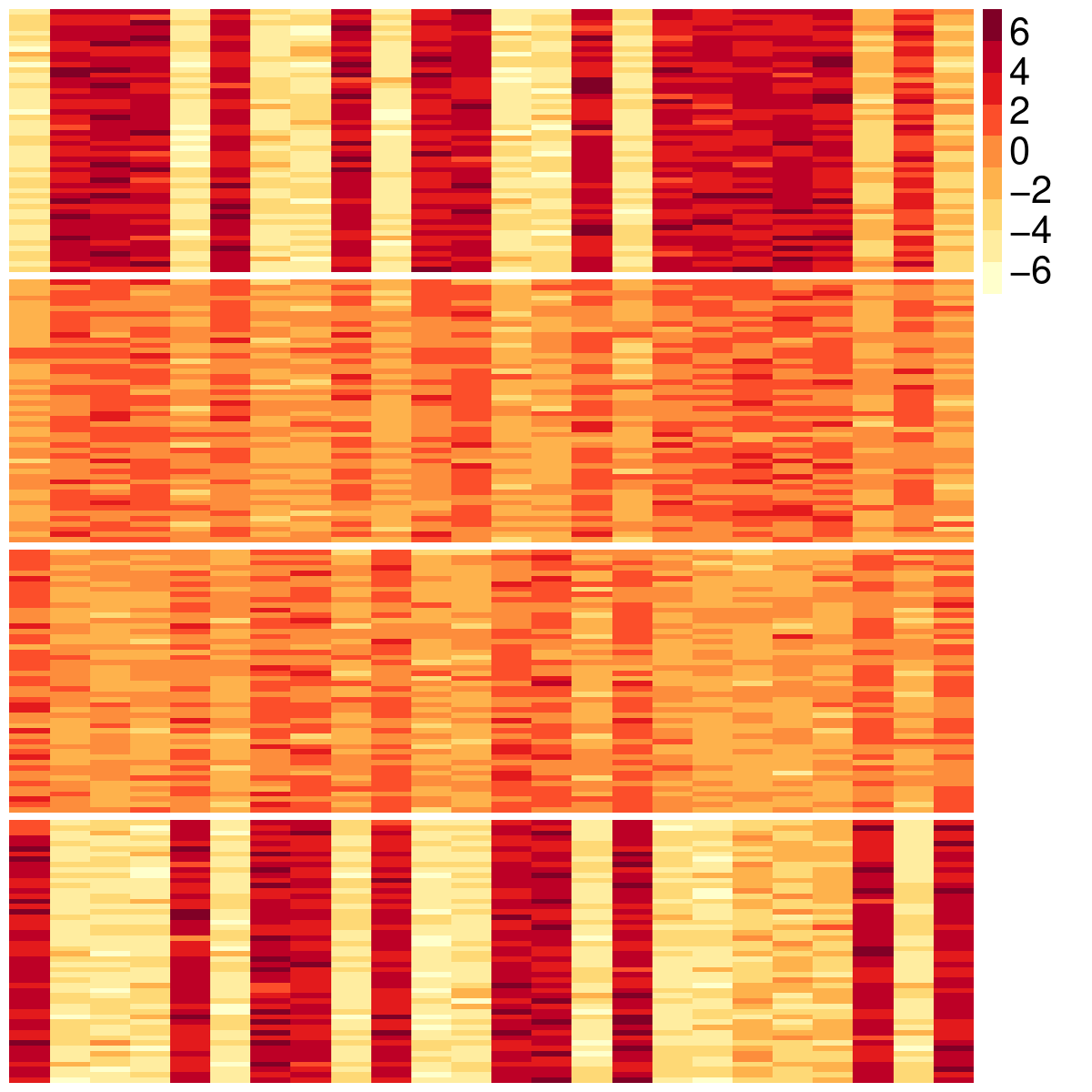}}
  \subfloat[$\mathcal G_4^*$]{
 \includegraphics[height=4cm, width =4cm]{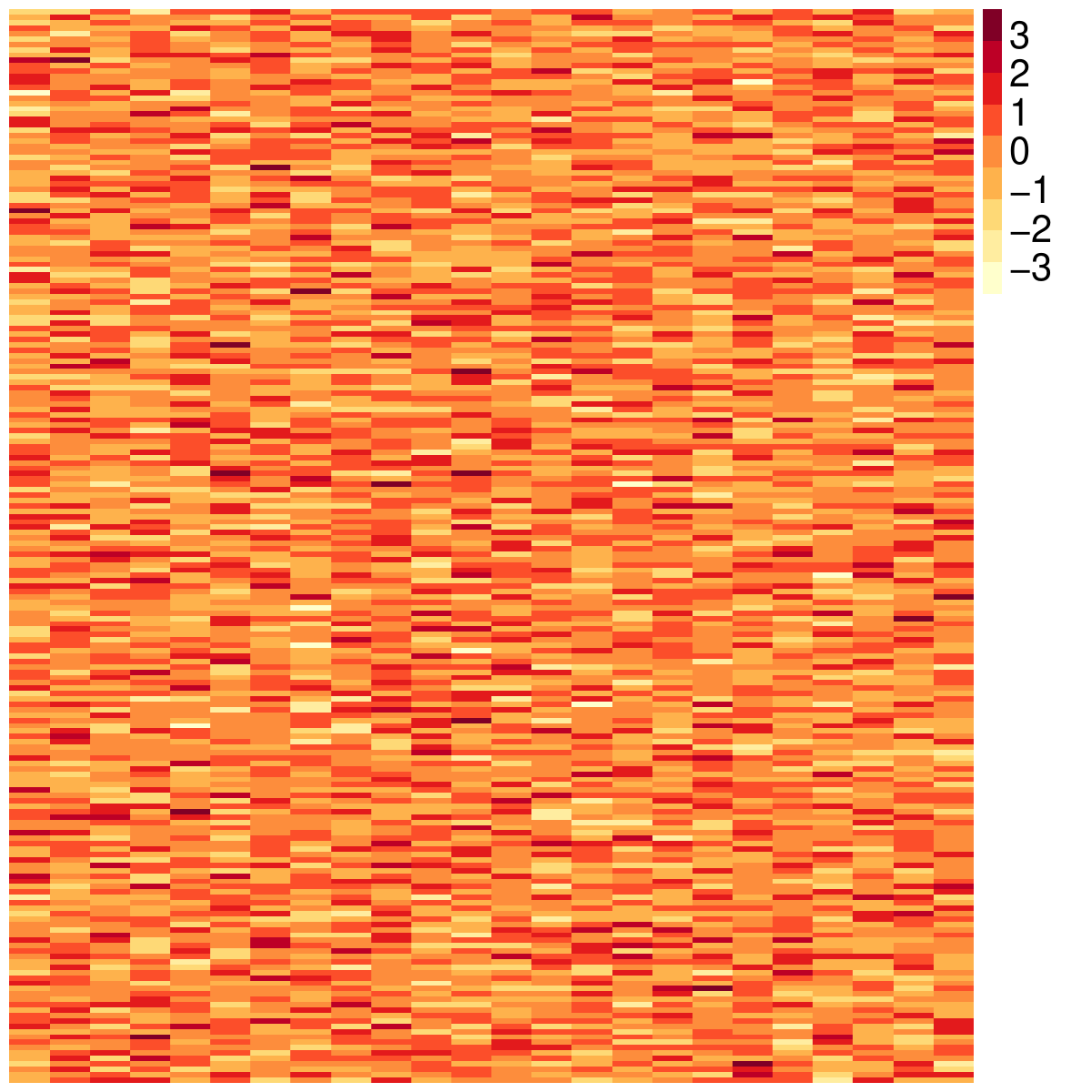}}\\
\centering
\caption{Heatmaps of the 4 feature groups with different clustering structures.}\label{heatmap_G3}
\end{figure}

For each setting, we generate 100 Monte Carlo replicates, except for the NBLC method, which is repeated 50 times due to its substantial computational burden as a Bayesian approach.
We evaluate the performance of each method from two perspectives: feature group partitioning and sample clustering. 
The accuracy of feature group selection is measured using two metrics: (a) fARI, the adjusted Rand index \citep{rand1971objective} between the true feature grouping and the estimated group partition, computed over all informative features; and (b) TNR (True Negative Rate), defined as the proportion of all non-informative features correctly assigned to the non-informative group.
The quality of sample clustering is assessed by three measures: (c) $\widehat M_g$, the average number of clusters identified within the $g$-th feature group; (d) sARI, the average adjusted Rand index between the true and estimated sample clustering structures across all informative feature groups; and (e) MAE, the average mean absolute error of the clustering matrices for informative features, defined as $\mathrm{MAE}=|\mathcal A^*|^{-1}\sum_{j\in\mathcal A^*}\|\widehat{\mathbf P}_{\widehat{g}_j}-\mathbf P_{g_j^*}^*\|_1$.

\subsection{Simulation results}
We provide a detailed illustration based on the simulation results. The performance of feature group partitioning is summarized in Figures \ref{g5} and \ref{g7}, while the sample clustering accuracy is reported in Table \ref{ex1_estimation} and Figures \ref{g6} and \ref{g8}. We begin by comparing the proposed LSC method with the Oracle benchmark to further validate our theoretical results. As shown in Figure \ref{g5}, the median fARI of the Oracle method remains at 1 due to its idealized design. 
For the LSC method, when $s=16$, the median fARI does not exhibit a clear improvement with increasing sample size. In contrast, for $s\in\{24,32\}$, the median fARI tend to 1 as the sample size increases, indicating that the true feature group structure can be consistently recovered when each group contains a sufficiently large number of features.
This finding is in line with Theorem 2. 
Turning to sample clustering performance, Figure \ref{g6} shows that the LSC method exhibits a similar pattern in terms of sARI. When the sample size is fixed (e.g. $n=100$), the median sARI of the LSC method increases substantially as $s$ rises from 16 to 24, but stabilizes when $s$ further increases to 32.
This pattern suggests that when the number of features is small (e.g., $s = 16$), clustering accuracy is primarily limited by insufficient feature information. However, when the number of features is larger (e.g., $s = 24$), the improvement in clustering accuracy becomes constrained by the smaller sample sizes, reflecting the trade-off characterized in Theorems 1 and 3. Consistent evidence is provided by Table \ref{ex1_estimation} and Figure \ref{g8}, which show that as both $s$ and $n$ increase, the performance of the LSC method converges to that of the Oracle method.

\begin{table}[H]
\centering\renewcommand\arraystretch{1}{\scriptsize
\renewcommand{\arraystretch}{0.6}
\caption{Estimated number of sample clusters within each informative feature group.
}\label{ex1_estimation}
\setlength{\tabcolsep}{3mm}{
\begin{tabular}{ccccccccccc}
			
	\hline
	\cline{1-11}
	&&\multicolumn{3}{c}{$s=16$}&\multicolumn{3}{c}{$s=24$}&\multicolumn{3}{c}{$s=32$}\\
	\cmidrule(lr){3-5}
	\cmidrule(lr){6-8}
	\cmidrule(lr){9-11}
	   &Method&$\widehat M_1$&$\widehat M_2$&$\widehat M_3$ &$\widehat M_1$&$\widehat M_2$&$\widehat M_3$ 
                 &$\widehat M_1$&$\widehat M_2$&$\widehat M_3$  \\
\hline
			
$n=100$   &LSC	     &2.07    &2.13    &3.91     &1.95    &2.90    &4.02    &1.94    &2.95   &4.00    \\
          &          &(0.39)  &(0.56)  &(0.45)   &(0.11)  &(0.45)  &(0.20)  &(0.05)  &(0.20) &(0.00)  \\
          &Oracle	   &2.00    &2.96    &4.00     &2.00    &3.00    &4.00    &2.00    &3.00   &4.00   \\
          &          &(0.00)  &(0.20)  &(0.00)   &(0.00)  &(0.00)  &(0.00)  &(0.00)  &(0.00) &(0.00)  \\
          &mLSC	   &1.53    &1.63    &1.80     &1.54    &1.65    &1.96    &1.56    &1.67   &2.05    \\
          &          &(0.15)  &(0.13)  &(0.41)   &(0.17)  &(0.15)  &(0.51)  &(0.17)  &(0.14) &(0.51)  \\
          
$n=200$   &LSC	     &1.99    &2.44    &4.02     &1.96    &2.96    &4.00    &1.98    &3.03   &4.02    \\
          &          &(0.34)  &(0.86)  &(0.49)   &(0.05)  &(0.20)  &(0.00)  &(0.03)  &(0.30) &(0.20)  \\
          &Oracle	   &2.00    &2.99    &4.00     &2.00    &3.00    &4.00    &2.00    &3.00   &4.00    \\
          &          &(0.00)  &(0.10)  &(0.00)   &(0.00)  &(0.00)  &(0.00)  &(0.00)  &(0.00) &(0.00)  \\
          &mLSC	   &1.47    &1.61    &1.70     &1.55    &1.58    &1.84    &1.56    &1.61   &1.85    \\
          &          &(0.19)  &(0.20)  &(0.43)   &(0.21)  &(0.12)  &(0.47)  &(0.23)  &(0.18) &(0.50)  \\
          
$n=400$   &LSC	     &2.06    &2.64    &4.08     &1.99    &2.97    &4.00    &1.99    &3.00   &4.00    \\
          &          &(0.44)  &(0.94)  &(0.49)   &(0.04)  &(0.17)  &(0.00)  &(0.01)  &(0.00) &(0.00)  \\
          &Oracle	   &2.00    &3.00    &4.00     &2.00    &3.00    &4.00    &2.00    &3.00   &4.00    \\
          &          &(0.00)  &(0.00)  &(0.00)   &(0.00)  &(0.00)  &(0.00)  &(0.00)  &(0.00) &(0.00)  \\
          &mLSC	   &1.53    &1.63    &1.74     &1.61    &1.60    &1.70    &1.57    &1.59   &1.66    \\
          &          &(0.29)  &(0.31)  &(0.43)   &(0.31)  &(0.23)  &(0.35)  &(0.30)  &(0.25) &(0.36)  \\
			\hline
			\cline{1-11}	
		\end{tabular}
	}}
\end{table}

We compare the LSC method with alternative approaches in terms of feature group selection and sample clustering performance. Overall, the NBLC method is the only competitor that achieves performance comparable to our proposed LSC method, while the remaining methods perform substantially worse. Specifically, the SKM and SCFS methods always exhibit poor performance under our settings, as both are designed for scenarios with a single homogeneous feature group ($G=1$). 
Although the multi-step designed mLSC method can be extended to cases with $G > 1$, it fails to effectively aggregate within-group feature signals, like LSC and NBLC do. Instead, it relies on individual feature-level information for both feature screening and sample clustering, resulting in poor performance in practice.
Comparing the LSC and NBLC methods, we find that the NBLC method performs slightly better when the sample size is small (e.g., $n=100$). However, it performs significantly worse than the LSC method when the number of features is not large enough (e.g., $s=16$ or 24) and the sample size is fixed at $n=200$ or 400. Moreover, as shown in Figure \ref{g8}, the NBLC method consistently underperforms the LSC method in terms of TNR. In addition, the Bayesian nature of NBLC incurs significantly higher computational cost compared to the proposed LSC method. Taken together, these results demonstrates that the LSC method outperforms the NBLC method and other competing methods either in accuracy or computational efficiency.

\begin{figure}[H]
\centering
\includegraphics[width =14cm]{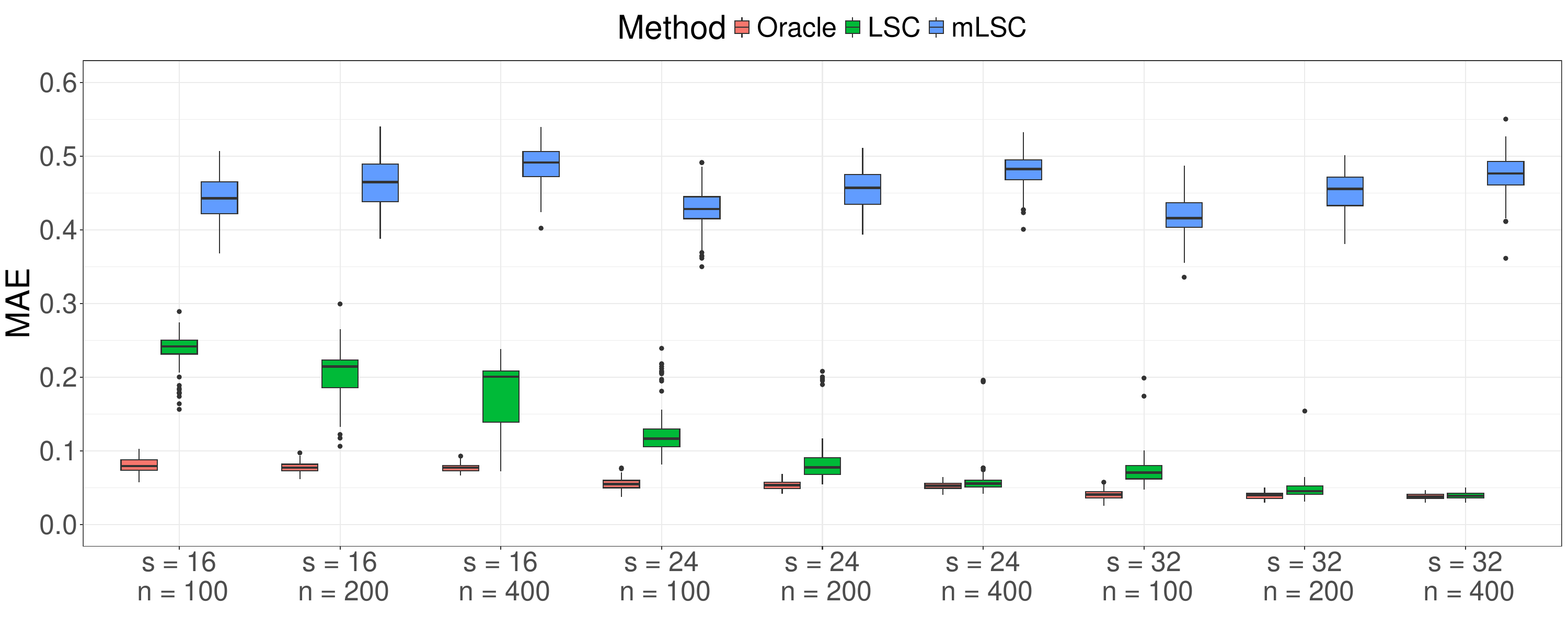}
\caption{Boxplots of MAE across methods under different settings.
}
\label{g8}
\end{figure}

\begin{figure}[H]
\centering
\includegraphics[width =15cm]{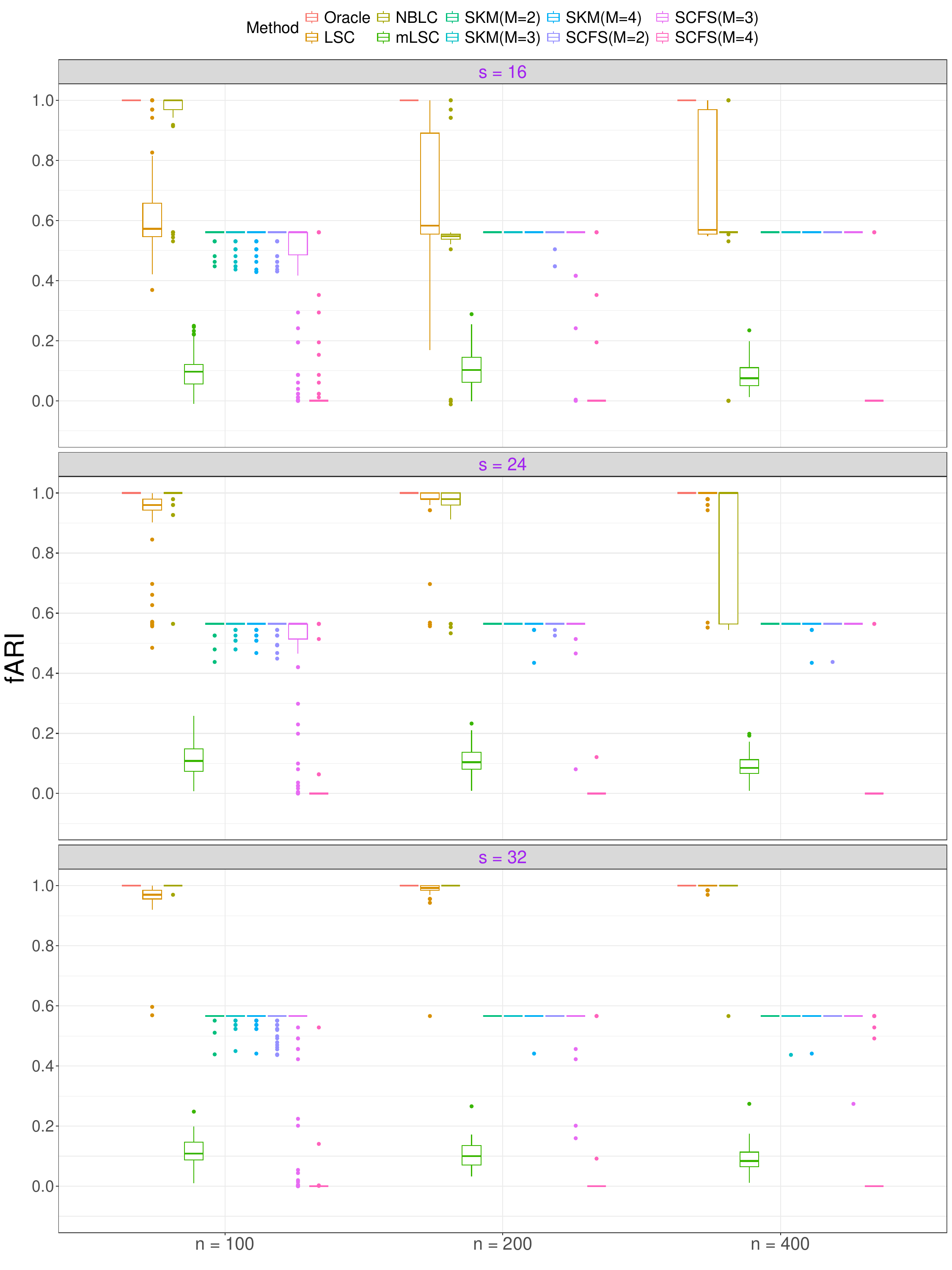}
\caption{Boxplots of fARI across methods under different settings.}
\label{g5}
\end{figure}

\begin{figure}[H]
\centering
\includegraphics[width =15cm]{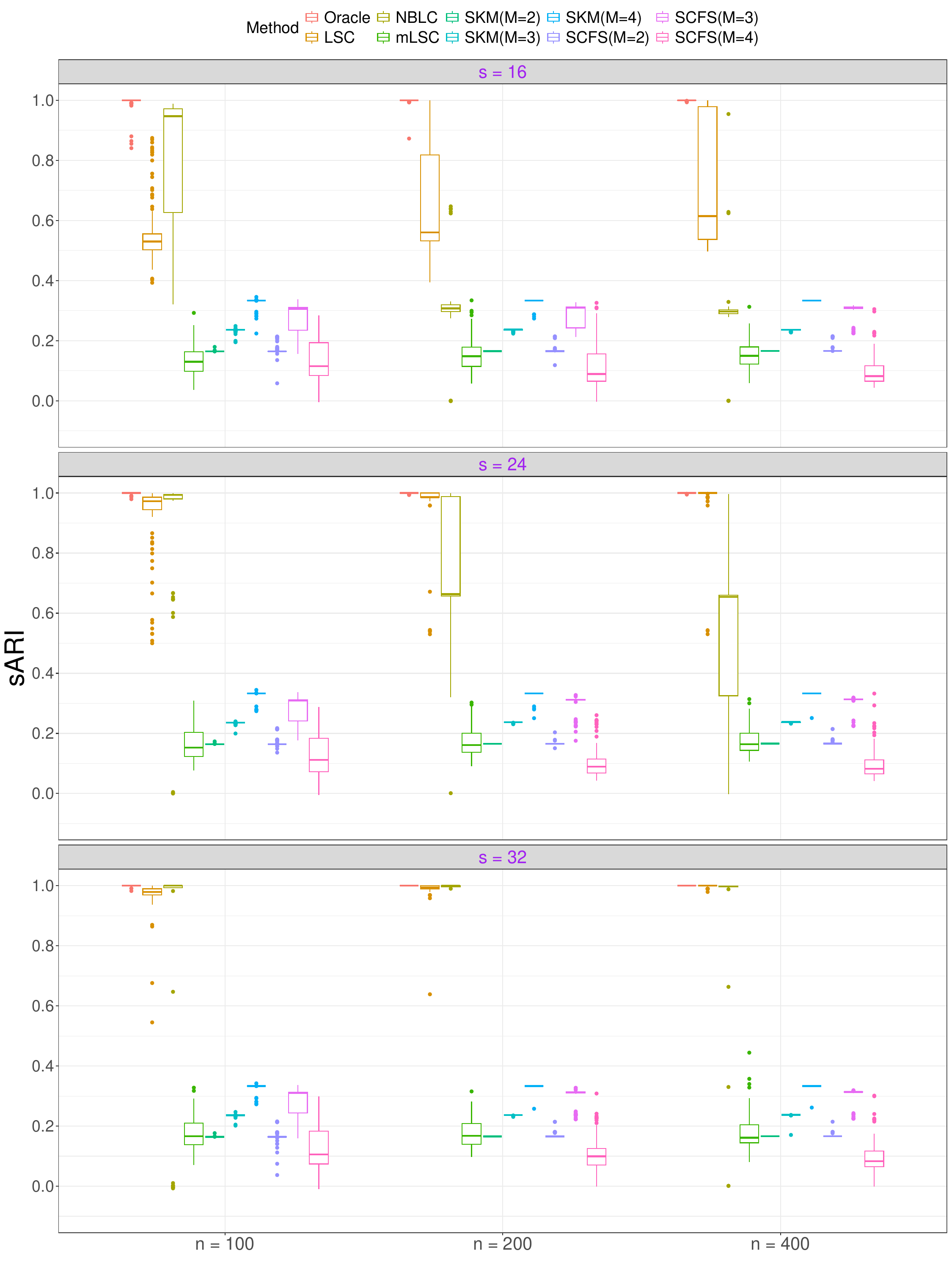}
\caption{Boxplots of sARI across methods under different settings.
}
\label{g6}
\end{figure}

\begin{figure}[H]
\centering
\includegraphics[width =15cm]{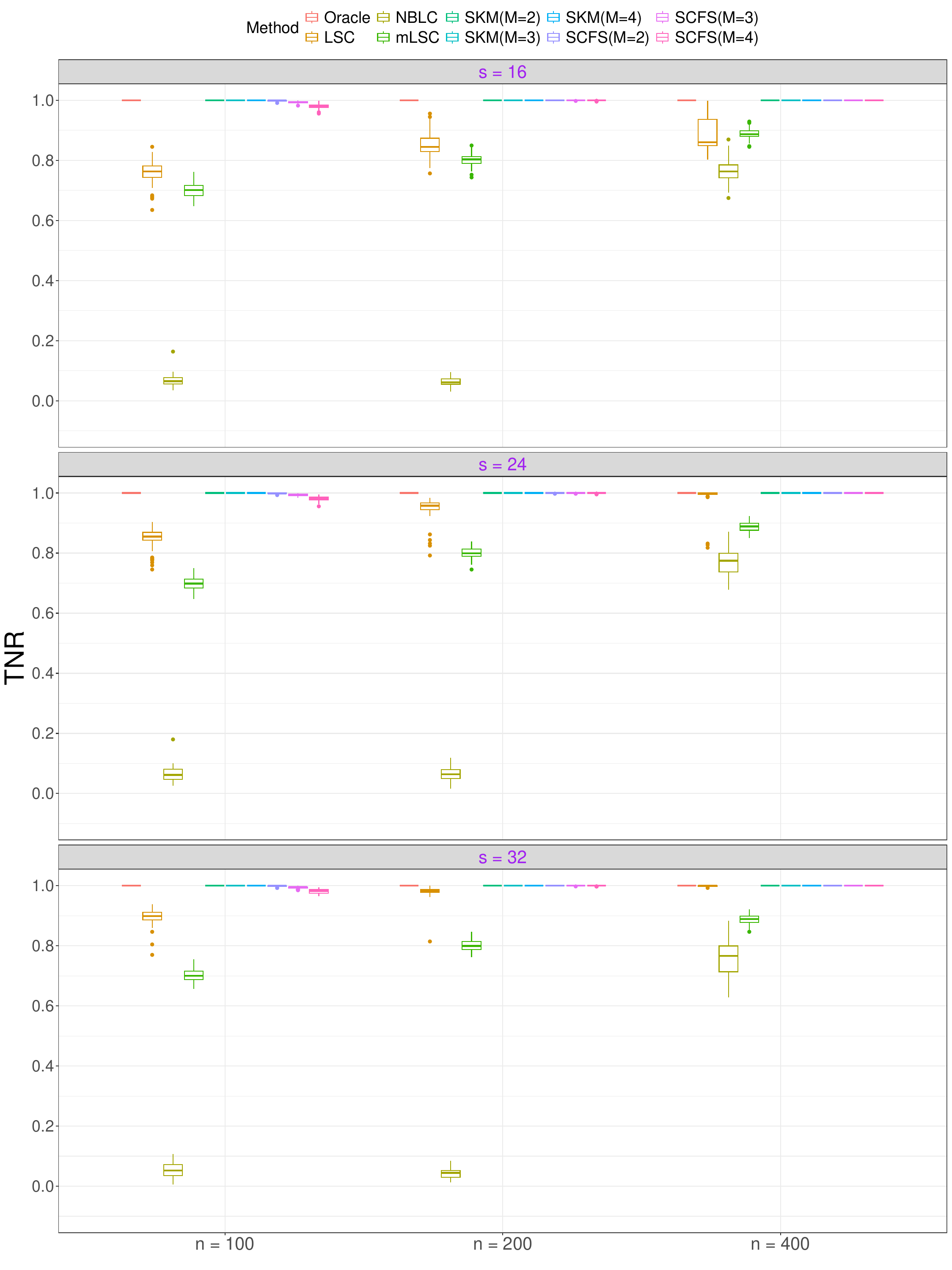}
\caption{Boxplots of TNR across methods under different settings..
}
\label{g7}
\end{figure}

\section{Real Data Analysis}
Acute Myeloid Leukemia (AML) remains a highly fatal malignancy, especially among elderly patients \citep{yilmaz2019late}. Since the 1970s, anthracycline combined with cytosine arabinoside (AraC), commonly referred to as conventional chemotherapy (CC), has been the standard of care in AML induction therapy \citep{yates1973cytosine}. More recently, increasing evidence has demonstrated that some patients with newly diagnosed AML benefit from the combination of venetoclax (VEN) and hypomethylating agents (HMA), referred to as VH therapy. However, due to the pronounced heterogeneity of AML, the benefits of VH therapy are not universal across patients \citep{de2024proteomics}. Therefore, accurately identifying patients who are more likely to benefit from one regimen over the other is crucial for maximizing outcomes with existing therapies.

In this section, we analyze Reverse-phase Protein Array (RPPA)-based proteomics for AML. This dataset, generated by \cite{de2024proteomics} and publicly available on GitHub (https://github.com/escmagalhaes/23-LEU-1445), contains protein expression measurements for 411 proteins across 146 patients.
In prior work, \cite{de2024proteomics} applied exploratory analyses to this dataset and stratified patients into two clusters based on Protein Set 3 (29 proteins). One cluster (61 patients) exhibited no significant benefit from either VH or CC therapy, whereas the other cluster of patients showed a markedly better response to CC compared to VH.
Based on this result, we have reason to suspect that Protein Set 3 does not contain sufficient discriminative information to help determine whether VH or CC is more effective in the first cluster.
Therefore, we apply our proposed method to re-partition the full set of 411 proteins. Our goal is to identify new protein group partition under which patients can be more finely stratified, thereby revealing significant differences in treatment response that were previously indistinguishable.

\begin{figure}[H]
\centering
\subfloat[Heatmap of the Protein Group 1.]{
 \includegraphics[height=8cm, width =8cm]{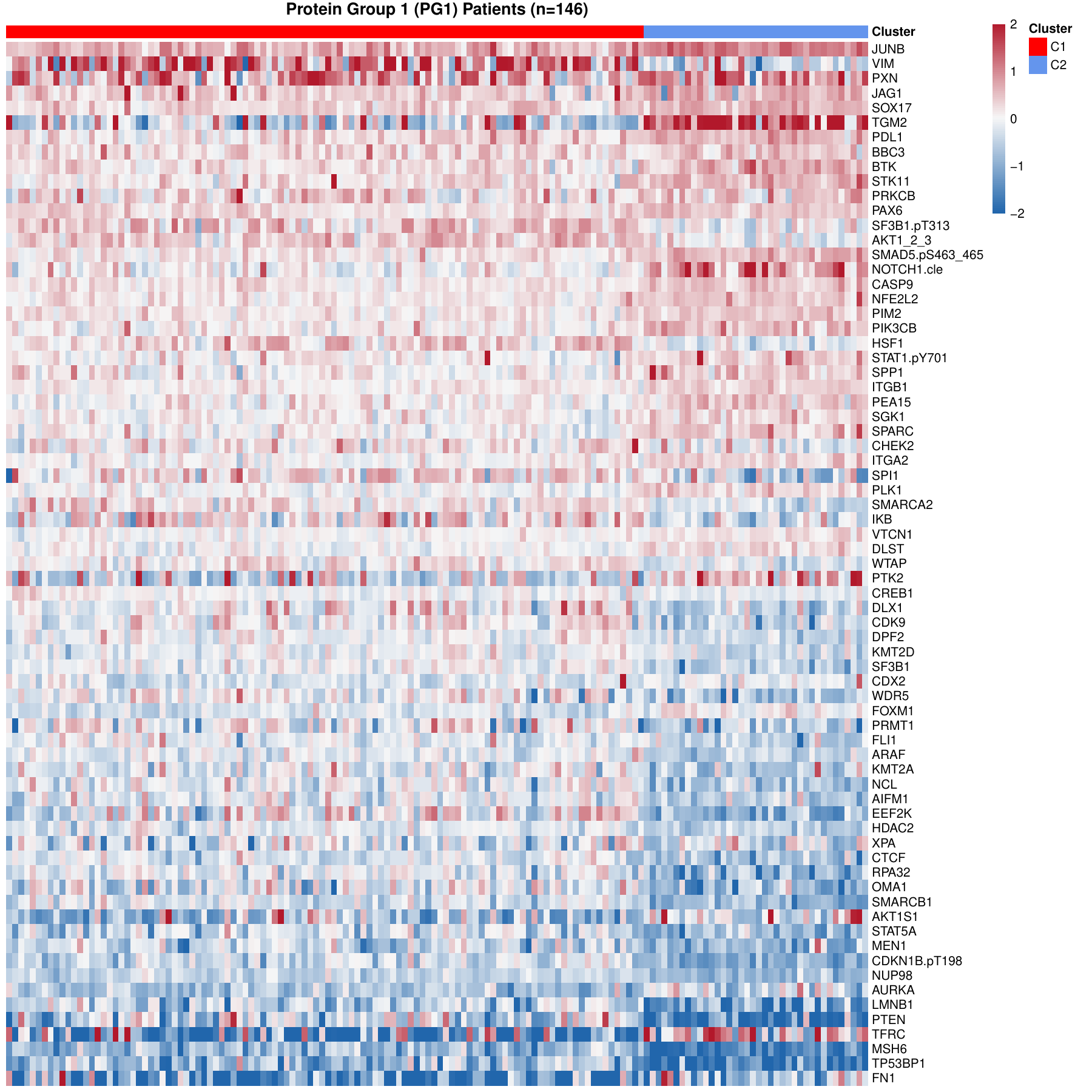}}
  \subfloat[Heatmap of the Protein Group 2.]{
 \includegraphics[height=8cm, width =8cm]{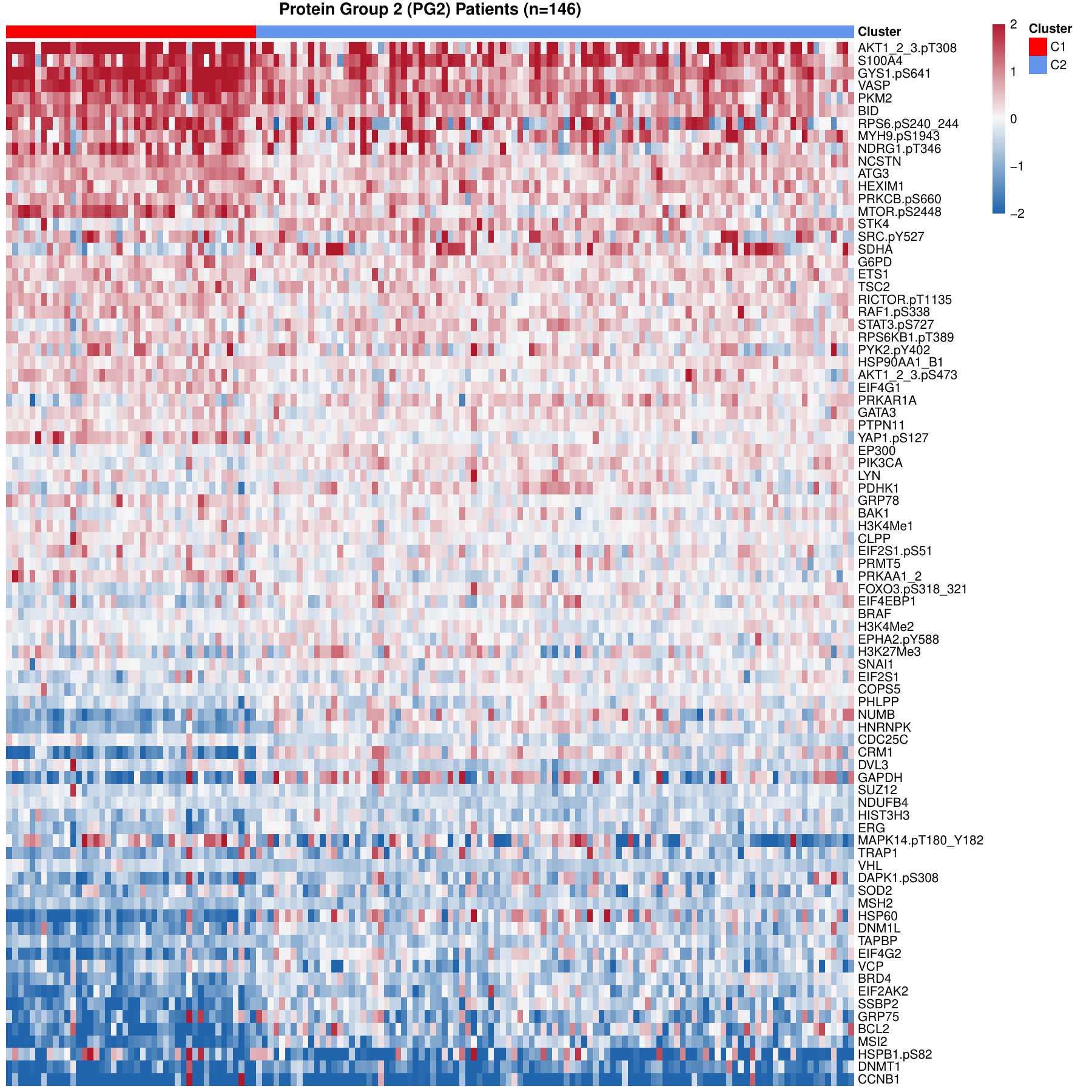}}\\
 \subfloat[Kaplan–Meier plots under Protein Group 1.  ]{
 \includegraphics[height=8cm, width =8cm]{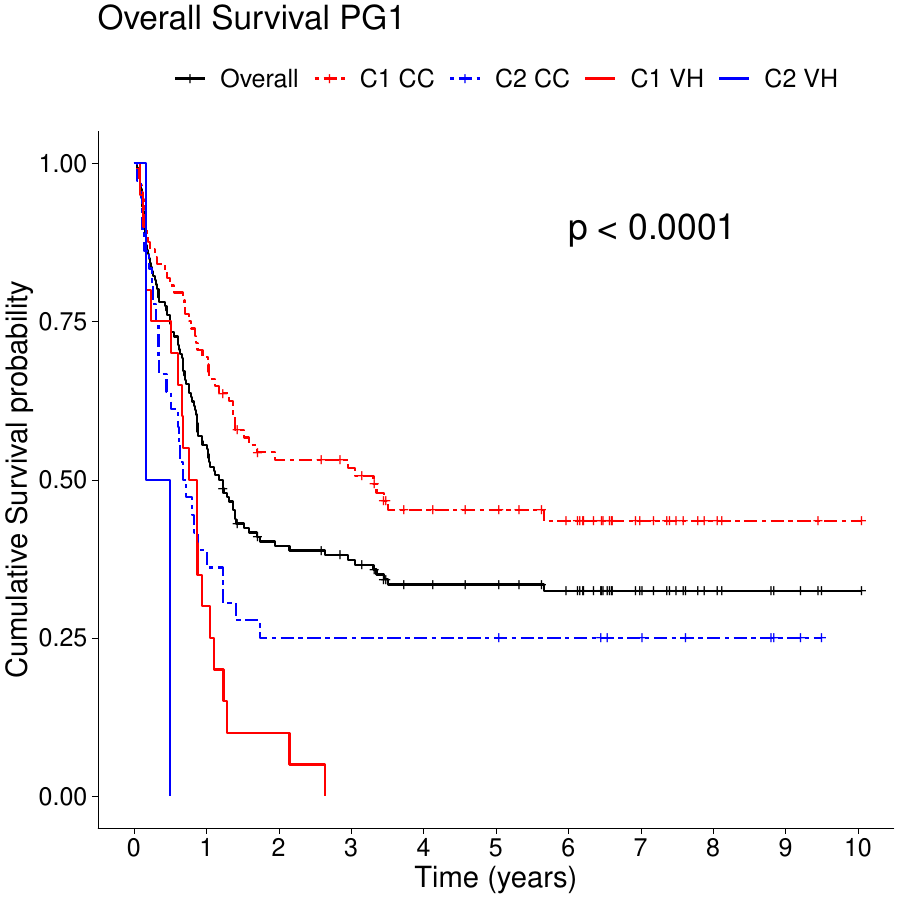}}
  \subfloat[Kaplan–Meier plots under Protein Group 2.]{
 \includegraphics[height=8cm, width =8cm]{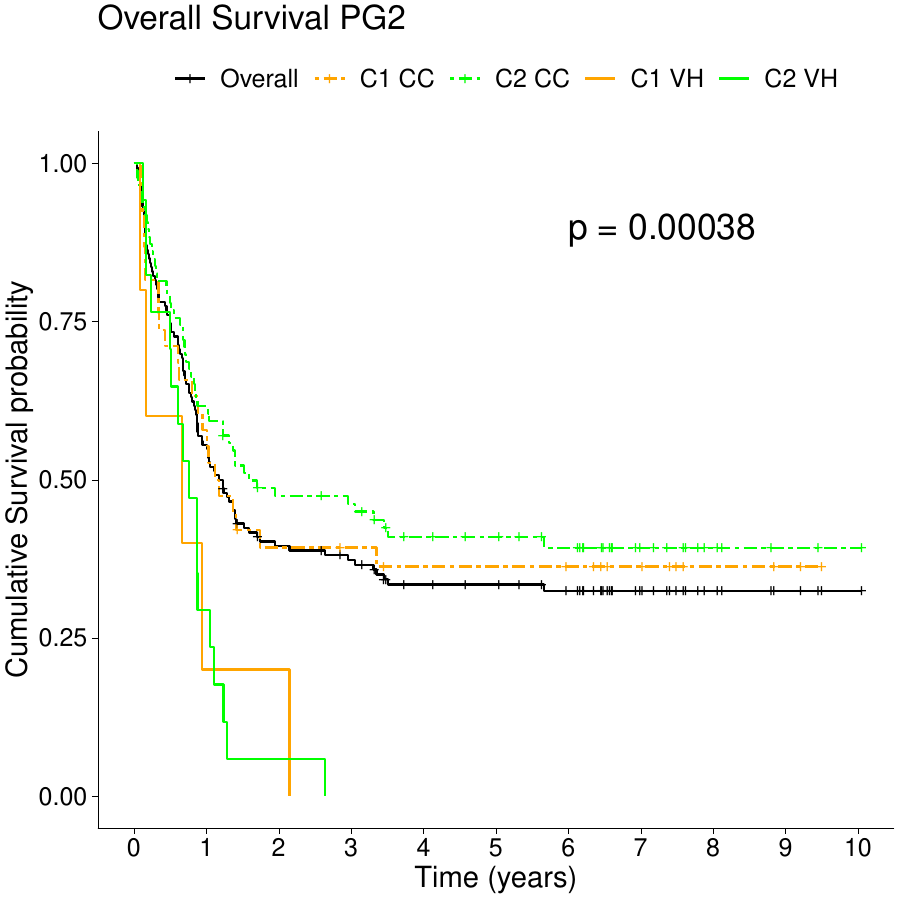}}\\
\centering
\caption{Heatmpas and KM plots for two identified protein groups.}\label{heatmap_kmcurve}
\end{figure}

Building on previous analysis, patients who exhibit differential responses to different treatment regimens may be identified through two distinct protein groups; therefore, we set $G=2$.
The application of our method partitions the 411 proteins into three groups: Protein Group 1 (PG1) with 71 proteins, Protein Group 2 (PG2) with 83 proteins, and Protein Group 3 (the non-informative group) with 257 proteins. 
Based on PG1 and PG2, respectively, the 146 patients can each be classified into two clusters. Figure \ref{heatmap_G3}(a) presents the heatmap of the identified proteins in PG1. Based on PG1, patients can be divided into two clusters: Cluster C1 (108 patients) and Cluster C2 (38 patients). 
Figure \ref{heatmap_G3}(c) shows the Kaplan–Meier plots for OS (Overall Survival), stratified by cluster and treatment under PG1.
As shown in Figure \ref{heatmap_G3}(c), patients in Cluster C1 (red) exhibit a favorable prognosis when treated with CC (dashed line), with an MS (median survival) of 39.7 months, but a very poor outcome when treated with VH (solid line), having an MS of 9.8 months. 
In contrast, patients in Cluster C2 (orange color) have similarly poor OS under both treatment regimens. This pattern closely resembles the survival curves reported in Figure 1(F) of \cite{de2024proteomics}. 
Notably, PG1 contains 9 proteins ({\it EEF2K, LMNB1, PDL1, PEA15, PIM2, PTEN, SMARCA2, SPI1, TP53BP1}) that exactly match those proteins in Protein Set 3 identified by \cite{de2024proteomics}. Furthermore, 82 of the 108 patients in Cluster C1 overlap with the C1 cluster (85 patients) shown in Figure 1(F) of \cite{de2024proteomics}.
Taken together, these results provide strong external validation for the biological and clinical relevance of the PG1 identified by our method.

Figure \ref{heatmap_G3}(b) shows the heatmap of the identified proteins in PG2. Based on PG2, the 146 patients can be divided into two clusters: Cluster C1 (43 patients) and Cluster C2 (103 patients). Figure \ref{heatmap_G3}(d) displays Kaplan–Meier curves for OS, stratified by cluster and treatment modality under PG2. 
As shown in Figure \ref{heatmap_G3}(d), patients in Cluster C1 (orange) had better survival prognosis when treated with CC (dashed line), with an MS of 13.8 months, compared to VH (solid line), which yields an MS of 8.0 months.
Similarly, patients in Cluster C2 (green) also showed a good prognosis when treated with CC (dashed line), with an MS of 19.1 months, and a poor outcome when treated with VH (solid line), having an MS of 9.1 months. 
Notably, in Figure 1(F) of \cite{de2024proteomics}, no significant difference between the two treatment approaches was observed within the corresponding cluster. In contrast, according to the results in Figure \ref{heatmap_G3}(d), our PG2-based clustering reveals clear treatment heterogeneity between patients receiving CC and those receiving VH within both clusters. 
In addition, PG2 only contains 4 proteins ({\it EIF2AK2, EIF4G1, H3K27Me3, HSPB1.pS82}), which are also included in Protein Set 3 of \cite{de2024proteomics}. Among the 61 patients classified as C2 in \cite{de2024proteomics}, 22 are assigned to C1 and 39 to C2 under our method. 
Taken together, these results indicate that our proposed method can identify a new set of proteins, under which a previously indistinguishable cluster in terms of treatment effects can be further divided into two new clusters with clearly differentiated treatment outcomes. This provides additional evidence for the advantage of our approach in capturing complex heterogeneity in AML proteomic profiles.

\section{Conclusion}\label{sec-conc}
In this paper, we develop a frequentist framework for local clustering in high-dimensional data with heterogeneous feature-specific clustering structures. In contrast to conventional clustering methods, which assume that all informative features support a common partition of the observations, our framework allows different subsets of features to induce distinct sample partitions. The proposed formulation, therefore, provides a flexible approach to analyzing data in which multiple clustering structures coexist across the feature space.

Several directions merit further investigation. First, the current framework assumes that each informative feature belongs to a single feature group and therefore supports one underlying sample partition. In some applications, however, a feature may be related to multiple biological pathways, clinical mechanisms, or latent data-generating processes \citep{ni2020bayesian}. Extending the proposed formulation to overlapping feature groups would allow a single feature to contribute to multiple clustering structures and would provide a more flexible representation of complex high-dimensional data.

Second, the present formulation focuses on feature groups that induce common sample partitions. A natural extension is to allow the partitions within a feature group to be similar but not necessarily identical. Such a relaxation would be useful when related features support approximately shared clustering structures but exhibit feature-specific deviations due to measurement error, weak secondary signals, or biological heterogeneity. Developing suitable distances or regularization schemes for clustering matrices may provide a principled way to quantify and estimate such approximate structural similarity.

Finally, further work is needed on the data-driven selection of the number of feature groups. In practice, this quantity may be unknown and may differ substantially across various datasets. Penalized optimization criteria, stability-based methods, cross-validation procedures, or information criteria could be incorporated into the proposed framework. Establishing the corresponding model selection consistency for the number of feature groups is particularly important for providing a fully adaptive local clustering procedure.

Overall, the proposed framework provides a new perspective on local clustering by explicitly treating heterogeneous sample partitions as estimable statistical structures. By jointly recovering feature groups and their associated sample partitions, it offers both a practical methodology and a theoretical foundation for identifying multiple clustering structures in modern high-dimensional data.

\bibliography{hma.bib}

{
\newpage
\section*{Appendix}

\renewcommand{\theequation}{A.\arabic{equation}}
\renewcommand\thelemma{A.\arabic{lemma}}
\renewcommand\thedefin{A.\arabic{defin}}
\setcounter{equation}{0}
\setcounter{lemma}{0}
\setcounter{defin}{0}

\section{Proofs}

\begin{defin}[M-means matrices] \label{def3}
    Let $\mathbb H^{n\times M}$ be the set of hard (cluster) labels: $\{0,1\}$-valued $n\times M$ matrices where each row has exactly a single 1. Then, we can define one class of M-means matrices as follows
    \[
    \begin{aligned}
    \mathbb M_{n,n'}^M:=&\ \left\{\mathbf X\in\mathbb R^{n\times n'}:\mathbf X\ \text{has at most M distinct rows}\right\}\\
    =&\ \left\{\mathbf Z\mathbf R: \mathbf Z\in\mathbb H^{n\times M}, \mathbf R\in\mathbb R^{M\times n'}\right\}.
    \end{aligned}
    \]
    Note that, there is a one-to-one correspondence between $\mathbf X\in\mathbb M_{n,n'}^M$ and $(\mathbf Z,\mathbf R)\in\mathbb H^{n\times M}\times\mathbb R^{M\times n'}$, up to label permutations. Thus, if $\mathbf X_1,\mathbf X_2\in\mathbb M_{n,n'}^M$ with the corresponding membership matrices $\mathbf Z_1,\mathbf Z_2\in\mathbb H^{n\times M}$, respectively, we define $\overline{\mathrm{Mis}}(\mathbf X_1,\mathbf X_2):=\overline{\mathrm{Mis}}(\mathbf Z_1,\mathbf Z_2)$.
\end{defin}

\begin{defin}[Center separation] \label{def4}
    For any $\mathbf X\in\mathbb M_{n,n'}^M$, we denote its distinct rows, referred to centers, as $\{q_m(\mathbf X),m\in[M]\}$ and define
    \[
    \delta_m(\mathbf X)=\min_{l:l\ne m}\ \|q_l(\mathbf X),q_m(\mathbf X)\|_2.
    \]
    In addition, let $n_m(\mathbf X)$ be the number of observations in cluster $m$ according to $\mathbf X$.
\end{defin}

\begin{lemma}\label{lemma1}
    If there exists a constant $0<\gamma<1$ such that $\min_{g\in[G]}\xi_g^{\mathrm{in}}>\gamma>\max_{g\in[G]}\xi_g^{\mathrm{out}}$, then taking $\tau=\gamma$ and the estimation error of $\widehat{\mathbf P}_{g}^{\mathrm{or}}$ can be bounded by
    \[
    \left\|\widehat{\mathbf P}_{g}^{\mathrm{or}}-\overline{\mathbf P}_g^{*}\right\|_1\le \frac{\left\langle{\mathbf R}_g-\tau\mathbf E_n,\overline{\mathbf P}_g^{*}-\widehat{\mathbf P}_{g}^{\mathrm{or}}\right\rangle}{\min\left\{\nu_g,1-\nu_g\right\}(\xi_g^{\mathrm{in}}-\xi_g^{\mathrm{out}})}, 
    \]
    where $\nu_g\in(0,1)$, and $\mathbf R_g$ is defined in \eqref{ref2} for $g\in[G]$. 
\end{lemma}
\noindent{\it Proof.}\quad The proof is similar to the proof of Lemma 2 in \cite{srivastava2023robust} and is therefore omitted here. $\hfill\square$

\begin{lemma}\label{lemma2}
    Recall the definition that ${\mathbf K}_{g}^{\mathrm{or}}=|{\mathcal A}_g^*|^{-1}\sum_{j\in{\mathcal A}_g^{*}}{\mathbf K}_j$ for each $g\in[G]$, then we have
    \[
    \left\langle{\mathbf R}_g-\tau\mathbf E_n,\overline{\mathbf P}_g^{*}-\widehat{\mathbf P}_{g}^{\mathrm{or}}\right\rangle\le 2\left\|{\mathbf K}_{g}^{\mathrm{or}}-{\mathbf R}_g\right\|_1. 
    \]
\end{lemma}
\noindent{\it Proof.}\quad The proof is similar to that of Lemma 3 in \cite{srivastava2023robust} and is therefore omitted here.$\hfill\square$

\begin{lemma}[\cite{zhou2019analysis}, Corollary 1]\label{lemma4}
    Let $\mathbf X^*\in\mathbb M_{n,n'}^M$ be a M-means matrix, and write $n_m=n_m(\mathbf X^*)$ and $\delta_m=\delta_m(\mathbf X^*)$. Assume that $\widehat{\mathbf X}\in\mathbb R^{n\times n'}$ satisfies that $\|\widehat{\mathbf X}-\mathbf X^*\|_F\le\zeta$ and
    \begin{enumerate}[(a)]
        \item $\mathbf X^*$ has exactly $M$ nonempty clusters;
        \item $c_m^{-2}(1+\omega)^2\zeta^2/(\delta_m^2n_m)<1$ for $m\in[M]$, and constants $c_m>0$ such that $c_m+c_{m'}\le 1$ for any $m\ne m'$.
    \end{enumerate}
    Then, any $\widetilde{\mathbf X}\in\mathcal P_{\omega}(\widehat{\mathbf X})$ has exactly $M$ clusters and
    \[
    \mathrm{Mis}_{m}(\mathbf X^*;\widetilde{\mathbf X})\le\frac{c_m^{-2}(1+\omega)^2\zeta^2}{n_m\delta_m^2},\quad\forall m\in[M].
    \]   
\end{lemma}

\noindent\textbf{Proof of Theorem 1.} 
Motivated by \cite{srivastava2023robust}, to show the results in Theorem \ref{thm1}, we should construct the corresponding reference matrices ${\mathbf R}_g$ such that the solution to \eqref{ref_center_obj} is some coarsening clustering matrix $\overline{\mathbf P}_g^{*}$, which corresponds to a coarsening of $\mathcal G_g^{*}$.
\begin{equation}\label{ref_center_obj}
\begin{aligned}
&\arg\max\ \left\langle{\mathbf R}_g-\tau\mathbf E_n,\mathbf P_g\right\rangle\\
&\text{subject to}\quad 0\le P_{g,ii'}\le 1\quad \text{for any}\  i,i'\in[n].
\end{aligned}
\end{equation}
Specifically, for $g\in[G]$, let ${\mathbf R}_g$ be a random matrix whose $(i,i')$-th element is defined as follows,
\begin{equation}\label{ref2}
 R_{g,ii'}:=\left\{
\begin{aligned}
&\max\left\{{K}_{g,ii'}^{\mathrm{or}},\ \xi_g^{\mathrm{in}} \right\} &\mathrm{if}&\quad i,i'\in\overline{\mathcal G}_{g,l}^{*} \ (l\in[L_g]), \\
&\min\left\{{K}_{g,ii'}^{\mathrm{or}},\ \xi_g^{\mathrm{out}} \right\} &\mathrm{if}&\quad i\in\overline{\mathcal G}_{g,l}^{*},i'\in\overline{\mathcal G}_{g,l'}^{*}\ (l\ne l'\in[L_g]). \\
\end{aligned}
\right.
\end{equation}
Here $\xi_g^{\mathrm{in}}$ and $\xi_g^{\mathrm{out}}$ are two thresholds defined for the within-cluster elements (distributed within diagonal blocks) and between-cluster elements (distributed within off-diagonal blocks), respectively. Note that, when $\xi_g^{\mathrm{in}}>\xi_g^{\mathrm{out}}$, there exists a $\gamma$ such that $\xi_g^{\mathrm{out}}<\gamma<\xi_g^{\mathrm{in}}$. Thus, taking $\tau=\gamma$ ensures the solution to \eqref{ref_center_obj} is some coarsening clustering matrix $\overline{\mathbf P}_g^{*}$.

According to Lemmas \ref{lemma1} and \ref{lemma2}, if we want to bound $\|\widehat{\mathbf P}_{g}^{\mathrm{or}}-\overline{\mathbf P}_g^{*}\|_1$, it suffices to bound $\|{\mathbf K}_{g}^{\mathrm{or}}-{\mathbf R}_g\|_1$.
Recall the definition of ${\mathbf R}^{(g)}$, we have 
\begin{equation}\label{overal_bound}
\begin{aligned}
    \left\|{\mathbf K}_{g}^{\mathrm{or}}-{\mathbf R}_g\right\|_1\le&\ \sum_{l\in[L_g]}\sum_{i,i'\in\overline{\mathcal G}_{g,l}^{*}}\mathbbm 1_{\{{K}_{g,ii'}^{\mathrm{or}}<\xi_g^{\text{in}}\}}\xi_g^{\text{in}}+\sum_{l\ne l'\in[L_g]}\sum_{i\in\overline{\mathcal G}_{g,l}^{*},i'\in\overline{\mathcal G}_{g,l'}^{*}}\mathbbm 1_{\{{K}_{g,ii'}^{\mathrm{or}}>\xi_{g}^{\text{out}}\}}(1-\xi_{g}^{\text{out}}) \\
    \le&\ \sum_{l\in[L_g]}I_g^{(l)}+\sum_{l\ne l'\in[L_g]}I_g^{(l,l')},
    \end{aligned}
\end{equation}
where $I_g^{(l)}$ and $I_g^{(l,l')}$ are the number of corruptions for the $l$-th diagonal block and $(l,l')$-th off-diagonal block, respectively. 
Recall that ${\mathbf K}_{g}^{\mathrm{or}}=s_g^{-1}\sum_{j\in{\mathcal A}_g^*}\mathbf K_j$ and $0<K_{j,ii'}\le 1$, then by Hoeffding's inequality,
\begin{equation}\label{hoefding}
\mathbb P\left[\left|{K}_{g,ii'}^{\mathrm{or}}-\mathbb E({K}_{g,ii'}^{\mathrm{or}})\right|\le t\right]\ge 1-2\exp\left\{-2s_gt^2\right\}.
\end{equation}
We let $\theta_j=\theta_{j'}=\vartheta_g$ for any $j,j'\in\mathcal A_g^*$.
When $i\in\mathcal G_{g,m}^{*},i'\in\mathcal G_{g,m'}^{*},m,m'\in\mathcal T_l,l\in[L_g]$, we have
\begin{equation}\label{inner_bound1}
\begin{aligned}
\mathbb E({K}_{g,ii'}^{\mathrm{or}})=&\ \mathbb E\left[s_g^{-1}\sum_{j\in\mathcal A_g^*}\exp\left\{-\frac{(\mu_{m,j}^{(g)}+\epsilon_{ij}-\mu_{m',j}^{(g)}-\epsilon_{i'j})^2}{2\vartheta_g^2}\right\}\right]\\
\ge&\ s_g^{-1}\sum_{j\in\mathcal A_g^*}\exp\left\{-\frac{\mathbb E(\mu_{m,j}^{(g)}-\mu_{m',j}^{(g)}+\epsilon_{ij}-\epsilon_{i'j})^2}{2\vartheta_g^2}\right\}\\
\ge&\ s_g^{-1}\sum_{j\in\mathcal A_g^*}\exp\left\{-\frac{(d_{\max,j}^{(g)})^2+2\sigma^2}{2\vartheta_g^2}\right\},\\
\end{aligned}
\end{equation}
where the first inequality follows from Jensen's inequality and the last inequality follows from the fact that $\epsilon_{ij}$ and $\epsilon_{i'j}$ are independent sub-Gaussian random variables with the same sub-Gaussian norm $\sigma$. Combining \eqref{hoefding} and \eqref{inner_bound1}, we can have
\[
\mathbb P\left[{K}_{g,ii'}^{\mathrm{or}}\ge s_g^{-1}\sum_{j\in\mathcal A_g^*}\exp\left\{-\frac{(d_{\max,j}^{(g)})^2+2\sigma^2}{2\vartheta_g^2}\right\}-t\right]\ge 1-2\exp\left\{-2s_gt^2\right\}.
\]
Moreover, let $t=(\kappa_{1,g}/s_g)\sum_{j\in\mathcal A_g^*}\exp\left\{-\frac{(d_{\max,j}^{(g)})^2+2\sigma^2}{2\vartheta_g^2}\right\}$ with $0<\kappa_{1,g}<1$ and then
\begin{equation}\label{inner_bound2}
\begin{aligned}
&\ \mathbb P\left[{K}_{\mathrm{or},ii'}^{(g)}\ge\frac{(1-\kappa_{1,g})}{s_g}\sum_{j\in\mathcal A_g^*}\exp\left\{-\frac{(d_{\max,j}^{(g)})^2+2\sigma^2}{2\vartheta_g^2}\right\}\right]\\
\ge&\ 1-2\exp\left\{-2[\kappa_{1,g}^2/s_g]\left[\sum_{j\in\mathcal A_g^*}\exp\left\{-\frac{(d_{\max,j}^{(g)})^2+2\sigma^2}{2\vartheta_g^2}\right\}\right]^2\right\}.
\end{aligned}
\end{equation}
Under \eqref{inner_bound2}, we let $\xi_g^{\mathrm{in}}=[(1-\kappa_{1,g})/s_g]\cdot\sum_{j\in\mathcal A_g^*}\exp\left\{-\frac{(d_{\max,j}^{(g)})^2+2\sigma^2}{2\vartheta_g^2}\right\}$ and then
\begin{equation}\label{inner_bound3}
\mathbb P\left[{K}_{g,ii'}^{\mathrm{or}}<\xi_g^{\mathrm{in}}\right]\le 2\exp\left\{-2(\kappa_{1,g}^2/s_g)\left[\sum_{j\in\mathcal A_g^*}\exp\left\{-\frac{(d_{\max,j}^{(g)})^2+2\sigma^2}{2\vartheta_g^2}\right\}\right]^2\right\}.
\end{equation}

Besides, when \(i\in{\mathcal G}_{g,m}^{*}\) and 
\(i'\in{\mathcal G}_{g,m'}^{*}\), where 
\(m\in\mathcal T_l\), \(m'\in\mathcal T_{l'}\), and \(l\ne l'\), we define
\[
\Delta_{mm',j}^{(g)}
=
\mu_{m,j}^{(g)}-\mu_{m',j}^{(g)},
\qquad
\eta_{ii',j}
=
\epsilon_{ij}-\epsilon_{i'j}.
\]
For each \(j\in\mathcal A_g^*\), define the events
\[
A_j=\left\{|\eta_{ii',j}|\le \frac{d_{\min,j}^{(g)}}{2}\right\},\quad \text{and}\quad
A_j^c=\left\{|\eta_{ii',j}|> \frac{d_{\min,j}^{(g)}}{2}\right\}.
\]

By the law of total expectation,
\begin{equation}
\label{eqa1}
\begin{aligned}
&\mathbb E\left[
\exp\left\{
-\frac{(\Delta_{mm',j}^{(g)}+\eta_{ii',j})^2}{2\vartheta_g^2}
\right\}
\right] \\
=&\
\mathbb E\left[
\exp\left\{
-\frac{(\Delta_{mm',j}^{(g)}+\eta_{ii',j})^2}{2\vartheta_g^2}
\right\}\Bigg\vert
A_j
\right]\mathbb P(A_j)
+
\mathbb E\left[
\exp\left\{
-\frac{(\Delta_{mm',j}^{(g)}+\eta_{ii',j})^2}{2\vartheta_g^2}
\right\}\Bigg\vert
A_j^c
\right]\mathbb P(A_j^c)\\
\le&\ \mathbb E\left[
\exp\left\{
-\frac{(\Delta_{mm',j}^{(g)}+\eta_{ii',j})^2}{2\vartheta_g^2}
\right\}\Bigg\vert
A_j
\right]+\mathbb P(A_j^c).
\end{aligned}
\end{equation}
On the event \(A_j\), by the triangle inequality,
\[
\left|\Delta_{mm',j}^{(g)}+\eta_{ii',j}\right|
\ge
\left|\Delta_{mm',j}^{(g)}\right|-|\eta_{ii',j}|
\ge
d_{\min,j}^{(g)}-\frac{d_{\min,j}^{(g)}}{2}
=
\frac{d_{\min,j}^{(g)}}{2},
\]
where the last inequality follows from $|\Delta_{mm',j}^{(g)}|\ge d_{\min,j}^{(g)}$.
Therefore,
\[
(\Delta_{mm',j}^{(g)}+\eta_{ii',j})^2\ge \frac{(d_{\min,j}^{(g)})^2}{4}.
\]
Hence,
\begin{equation}
\label{eqa2}
\mathbb E\left[
\exp\left\{
-\frac{(\Delta_{mm',j}^{(g)}+\eta_{ii',j})^2}{2\vartheta_g^2}
\right\}\Bigg\vert
A_j
\right]
\le
\exp\left\{
-\frac{(d_{\min,j}^{(g)})^2}{8\vartheta_g^2}
\right\}.
\end{equation}

By the definition of the sub-Gaussian norm, $\eta_{ii',j}$ is also a sub-Gaussian variable with sub-Gaussian norm $\sqrt{2}\sigma$. Then we have
\begin{equation}
\label{eqa3}
\mathbb P\left(|\eta_{ii',j}|> \frac{d_{\min,j}^{(g)}}{2}\right)\le 2\exp\left\{-\frac{(d_{\min,j}^{(g)})^2}{16\sigma^2}\right\}.
\end{equation}
Combining \eqref{eqa1} -- \eqref{eqa3},
\begin{equation}
\label{eqa4}
\mathbb E\left[
\exp\left\{
-\frac{(\Delta_{mm',j}^{(g)}+\eta_{ii',j})^2}{2\vartheta_g^2}
\right\}
\right] 
\le \exp\left\{
-\frac{(d_{\min,j}^{(g)})^2}{8\vartheta_g^2}
\right\}+2\exp\left\{-\frac{(d_{\min,j}^{(g)})^2}{16\sigma^2}\right\}
\end{equation}

Recall that
\[
K_{g,ii'}^{\mathrm{or}}
=
\frac{1}{s_g}
\sum_{j\in\mathcal A_g^*}
\exp\left\{
-\frac{(\Delta_{mm',j}^{(g)}+\eta_{ii'j})^2}{2\vartheta_g^2}
\right\},
\]
by \eqref{eqa4}, we have
\begin{equation}
\label{eqa5}
\mathbb E(K_{g,ii'}^{\mathrm{or}})
\le
\frac{1}{s_g}
\sum_{j\in\mathcal A_g^*}
\left[
\exp\left\{
-\frac{(d_{\min,j}^{(g)})^2}{8\vartheta_g^2}
\right\}+2\exp\left\{-\frac{(d_{\min,j}^{(g)})^2}{16\sigma^2}\right\}
\right].
\end{equation}

Combining \eqref{hoefding} and \eqref{eqa5},
\[
\mathbb P\left[{K}_{g,ii'}^{\mathrm{or}}\le s_g^{-1}\sum_{j\in\mathcal A_g^*}\left(
\exp\left\{
-\frac{(d_{\min,j}^{(g)})^2}{8\vartheta_g^2}
\right\}+2\exp\left\{-\frac{(d_{\min,j}^{(g)})^2}{16\sigma^2}\right\}
\right)+t\right]\ge 1-2\exp\left\{-2s_gt^2\right\}.
\]
Similarly, let 
\[
t=(\kappa_{1,g}/s_g)\sum_{j\in\mathcal A_g^*}\left(
\exp\left\{
-\frac{(d_{\min,j}^{(g)})^2}{8\vartheta_g^2}
\right\}+2\exp\left\{-\frac{(d_{\min,j}^{(g)})^2}{16\sigma^2}\right\}
\right),
\]
and then
\begin{equation}\label{inner_bound5}
\begin{aligned}
&\ \mathbb P\left[{K}_{g,ii'}^{\mathrm{or}}\le\frac{1+\kappa_{1,g}}{s_g}\sum_{j\in\mathcal A_g^*}\left(
\exp\left\{
-\frac{(d_{\min,j}^{(g)})^2}{8\vartheta_g^2}
\right\}+2\exp\left\{-\frac{(d_{\min,j}^{(g)})^2}{16\sigma^2}\right\}
\right)\right]\\
\ge&\ 1-2\exp\left\{-2(\kappa_{1,g}^2/s_g)\left[\sum_{j\in\mathcal A_g^*}\left(
\exp\left\{
-\frac{(d_{\min,j}^{(g)})^2}{8\vartheta_g^2}
\right\}+2\exp\left\{-\frac{(d_{\min,j}^{(g)})^2}{16\sigma^2}\right\}
\right)\right]^2\right\}.
\end{aligned}
\end{equation}
Under \eqref{inner_bound5}, we let $\xi_{g}^{\mathrm{out}}=[(1+\kappa_{1,g})/s_g]\cdot\sum_{j\in\mathcal A_g^*}\left(
\exp\left\{
-\frac{(d_{\min,j}^{(g)})^2}{8\vartheta_g^2}
\right\}+2\exp\left\{-\frac{(d_{\min,j}^{(g)})^2}{16\sigma^2}\right\}
\right)$ and then
\begin{equation}\label{inner_bound6}
\mathbb P\left[{K}_{g,ii'}^{\mathrm{or}}>\xi_{g}^{\mathrm{out}}\right]\le 2\exp\left\{-2(\kappa_{1,g}^2/s_g)\left[\sum_{j\in\mathcal A_g^*}\left(
\exp\left\{
-\frac{(d_{\min,j}^{(g)})^2}{8\vartheta_g^2}
\right\}+2\exp\left\{-\frac{(d_{\min,j}^{(g)})^2}{16\sigma^2}\right\}
\right)\right]^2\right\}.
\end{equation}

Next, we let $p_{g,l}=\mathbb P[{K}_{g,ii'}^{\mathrm{or}}<\xi_g^{\mathrm{in}}\big|i,i'\in\overline{\mathcal G}_{g,l}^{*}]$. Note that
\[
U_{g,l}:=\frac{\sum_{\{(i,i'):i<i'\in\overline{\mathcal G}_{g,l}^{*}\}}\mathbbm 1_{\{{K}_{g,ii'}^{\mathrm{or}}<\xi_g^{\text{in}}\}}}{\overline n_{l}^{(g)}(\overline n_{l}^{(g)}-1)/2}
\]
is an unbiased estimator for $p_{g,l}$ with $\overline n_l^{(g)}:=|\overline{\mathcal G}_{g,l}^{*}|$. By Bernstein's inequality for one-sample U-statistic \citep{arcones1995bernstein},
\begin{equation}\label{onesample_U1}
    \mathbb P\left[U_{g,l}-p_{g,l}>t_1\right]\le\exp\left\{-\frac{(\overline n_{l}^{(g)}/2)t_1^2}{c_1\nu_{g,l}+c_2t_1}\right\},
\end{equation}
where $\nu_{g,l}$ is the variance of the indicator random variable $\mathbbm 1_{\{{K}_{g,ii'}^{\mathrm{or}}<\xi_g^{\text{in}}\}}$ with $i,i'\in\overline{\mathcal G}_{g,l}^{*}$ and $c_1,c_2>0$ are constants. Taking $t_1=\max\{p_{g,l},c_3\log \overline n_{\min}^{(g)}/\overline n_{\min}^{(g)}\}$ where $c_3=2(c_1+c_2)$ and $\overline n_{\min}^{(g)}:=\min_{l\in[L_g]}\overline n_l^{(g)}$. Note that $\nu_{g,l}=p_{g,l}(1-p_{g,l})\le p_{g,l}\le t_1$, we can simplify \eqref{onesample_U1} to
\begin{equation}\label{onesample_U2}
    \mathbb P\left[U_{g,l}-p_{g,l}>t_1\right]\le\exp\left\{-\frac{\overline n_{l}^{(g)}t_1}{c_3}\right\}.
\end{equation}
Thus, with probability at least $1-1/\overline n_{\min}^{(g)}$,
\begin{equation}\label{onesample_U3}
\begin{aligned}
    I_g^{(l)}\le&\ 2\times\frac{\overline n_l^{(g)}(\overline n_l^{(g)}-1)}{2}\times\left[p_{g,l}+\max\left\{p_{g,l},c_3\log \overline n_{\min}^{(g)}/\overline n_{\min}^{(g)}\right\}\right]\\
    \le&\ 2\max\left\{p_{g,l},c_3\log \overline n_{\min}^{(g)}/\overline n_{\min}^{(g)}\right\}(\overline n_l^{(g)})^2.
\end{aligned}
\end{equation}
Similarly, let $p_{g,ll'}=\mathbb P[{K}_{g,ii'}^{\mathrm{or}}>\xi_g^{\mathrm{out}}\big|i\in\overline{\mathcal G}_{g,l}^{*},i'\in\overline{\mathcal G}_{g,l'}^{*}]$ and then 
\[
U_{g,ll'}:=\frac{\sum_{\{(i,i'):i\in\overline{\mathcal G}_{g,l}^{*},i'\in\overline{\mathcal G}_{g,l'}^{*}\}}\mathbbm 1_{\{{K}_{g,ii'}^{\mathrm{or}}>\xi_g^{\text{out}}\}}}{\overline n_{l}^{(g)}\overline n_{l'}^{(g)}}
\]
is also a U-statistic for $p_{g,ll'}$. By the simplified Bernstein's inequality for two-sample U-statistic ((E.14) in \cite{srivastava2023robust}), we have
\begin{equation}\label{twosample_U1}
        \mathbb P\left[U_{g,ll'}-p_{g,ll'}>t_2\right]\le\exp\left\{-\frac{\min\{\overline n_l^{(g)},\overline n_{l'}^{(g)}\}t_2^2}{c_4\nu_{g,ll'}+c_5t_2}\right\},
\end{equation}
where $\nu_{g,ll'}$ is the variance of the indicator random variable $\mathbbm 1_{\{{K}_{g,ii'}^{\mathrm{or}}>\xi_g^{\text{out}}\}}$ with $i\in\overline{\mathcal G}_{g,l}^{*},i'\in\overline{\mathcal G}_{g,l'}^{*}$, and $c_4,c_5>0$ are constants. Putting $t_2=\max\{p_{g,ll'},c_6\log \overline n_{\min}^{(g)}/\overline n_{\min}^{(g)}\}$ where $c_6=2(c_4+c_5)$ and noting that $\nu_{g,ll'}=p_{g,ll'}(1-p_{g,ll'})\le p_{g,ll'}\le t_2$, we have that with probability at least $1-1/(\overline n_{\min}^{(g)})^2$,
\begin{equation}\label{twosample_U2}
\begin{aligned}
    I_g^{(l,l')}\le&\ \overline n_l^{(g)}\overline n_{l'}^{(g)}\times\left[p_{g,ll'}+\max\left\{p_{g,ll'},c_6\log \overline n_{\min}^{(g)}/\overline n_{\min}^{(g)}\right\}\right]\\
    \le&\ 2\max\left\{p_{g,ll'},c_6\log \overline n_{\min}^{(g)}/\overline n_{\min}^{(g)}\right\}\overline n_l^{(g)}\overline n_{l'}^{(g)}.
\end{aligned}
\end{equation}
Combining Lemmas \ref{lemma1}, \ref{lemma2}, and the union bound, as well as \eqref{overal_bound}, \eqref{onesample_U3}, and \eqref{twosample_U2}, let $\rho_{g}=\min\left\{\nu_g,1-\nu_g\right\}(\xi_g^{\mathrm{in}}-\xi_g^{\mathrm{out}})$ and we have that with probability at least $1-2M_{\max}/n_{\min}$,
\begin{equation}\label{final_bound1}
\left\|\widehat{\mathbf P}_{g}^{\mathrm{or}}-\overline{\mathbf P}_g^{*}\right\|_1\le \frac{2}{\rho_g}\left\|{\mathbf K}_{g}^{\mathrm{or}}-{\mathbf R}_g\right\|_1\le\frac{4n^2}{\rho_g}\cdot\max\left\{\max_{l\ne l'\in[L_g]}p_{g,ll'},\max_{l\in[L_g]}p_{g,l},\frac{C_1\log n_{\min}}{n_{\min}}\right\},
\end{equation}
where $C_1=\max\{c_3,c_6\}$ and $n_{\min}=\min_{g\in[G],m\in[M_g]}n_m^{(g)}$. By \eqref{inner_bound3} and \eqref{inner_bound6},
\begin{equation}\label{final_bound2}
\begin{aligned}
\max_{l\ne l'\in[L_g]}p_{g,ll'}\le&\ 2\exp\left\{-2(\kappa_{1,g}^2s_g)\left[\frac{1}{s_g}\sum_{j\in\mathcal A_g^*}\left(
\exp\left\{
-\frac{(d_{\min,j}^{(g)})^2}{8\vartheta_g^2}
\right\}+2\exp\left\{-\frac{(d_{\min,j}^{(g)})^2}{16\sigma^2}\right\}
\right)\right]^2\right\},\\
\max_{l\in[L_g]}p_{g,l}\le&\ 2\exp\left\{-2(\kappa_{1,g}^2s_g)\left[\frac{1}{s_g}\sum_{j\in\mathcal A_g^*}\exp\left\{-\frac{(d_{\max,j}^{(g)})^2+2\sigma^2}{2\vartheta_g^2}\right\}\right]^2\right\}.\\
\end{aligned}
\end{equation}
Next, we show that $\rho_{g}>0$ when Condition \ref{con_snr1} holds for some $g\in[G]$. To see this, it suffices to show that $\xi_g^{\mathrm{in}}-\xi_g^{\mathrm{out}}>0$. Since $\xi_g^{\mathrm{in}}>0$ and $\xi_g^{\mathrm{out}}>0$, Condition \ref{con_snr1} implies 
\[
\begin{aligned}
\frac{\xi_g^{\mathrm{in}}}{\xi_g^{\mathrm{out}}}=&\ \left([(1-\kappa_{1,g})/s_g]\cdot\sum_{j\in\mathcal A_g^*}\exp\left\{-\frac{(d_{\max,j}^{(g)})^2+2\sigma^2}{2\vartheta_g^2}\right\}\right)\\
&\ \left([(1+\kappa_{1,g})/s_g]\cdot\sum_{j\in\mathcal A_g^*}\left[
\exp\left\{
-\frac{(d_{\min,j}^{(g)})^2}{8\vartheta_g^2}
\right\}+2\exp\left\{-\frac{(d_{\min,j}^{(g)})^2}{16\sigma^2}\right\}
\right]\right)^{-1}\\
>&\ 1.\\
\end{aligned}
\]
In \eqref{final_bound1}, note that there exists a constant $C_s>0$ such that
\[
    \left[\frac{1}{s_g}\sum_{j\in\mathcal A_g^*}\left(
\exp\left\{
-\frac{(d_{\min,j}^{(g)})^2}{8\vartheta_g^2}
\right\}+2\exp\left\{-\frac{(d_{\min,j}^{(g)})^2}{16\sigma^2}\right\}
\right)\right]^2\le C_s,
\]
\[
\left[\frac{1}{s_g}\sum_{j\in\mathcal A_g^*}\exp\left\{-\frac{(d_{\max,j}^{(g)})^2+2\sigma^2}{2\vartheta_g^2}\right\}\right]^2\le C_s
\]

Based on this result, we combine \eqref{final_bound1}, \eqref{final_bound2} and Condition \ref{con_snr1}, after taking $C=\max\{2,C_1\}\times 4/\rho_g$, then with probability at least $1-2M_{\max}/n_{\min}$
\begin{equation}\label{thm1_bound1}
\left\|\widehat{\mathbf P}_{g}^{\mathrm{or}}-\overline{\mathbf P}_g^{*}\right\|_1\le Cn^2\cdot\max\left\{\exp\left\{-\kappa s_g\right\},\frac{\log n_{\min}}{n_{\min}}\right\},
\end{equation}
where $\kappa>0$ is a constant depending on $\{\kappa_{1,g},\kappa_{2,g}\}_{g=1}^G$. We finish the proof of Theorem \ref{thm1} by applying the union bound over $G$ feature groups. $\hfill\square$

\newpage
\noindent\textbf{Proof of Theorem 2.}
To show the results in Theorem 2, we consider the neighborhood of $\mathbf P^*$ as follows,
\[
\bm\Theta=\left\{\mathbf P\in[0,1]^{n\times nG}:\sup_{g\in[G]}\left\|{\mathbf P}_g-\mathbf P_g^*\right\|_1\le r_n\right\}.
\]
Recall the definition of \eqref{obj}, let
\begin{equation}\label{objj}
\mathcal L(\mathbf P,\mathbf g)=\sum_{j=1}^p\left\langle\mathbf K_j-\tau\mathbf E_n,\mathbf P_{g_j}\right\rangle.
\end{equation}
Under $\bm\Theta$, given $g\in[G+1]$, if there exists some $j\in\mathcal A_g^*$ such that we can get a larger $\mathcal L(\mathbf P,\mathbf g)$ in \eqref{objj} by assigning it from its true group $\mathcal A_g^*$ to any other group $\mathcal A_{g'}^*$ with $g\ne g'$, then the oracle estimator $\widehat{\mathbf P}^{\mathrm{or}}$ is no longer the strictly local maximizer of $\mathcal L(\mathbf P,\mathbf g)$. In other words, if we can show that no such $j\in\mathcal A_g^*$ exists for any $g\in[G+1]$ with high probability, then $\widehat{\mathbf P}^{\mathrm{or}}$ is the strictly local maximizer of $\mathcal L(\mathbf P,\mathbf g)$ for any $\mathbf P\in\bm\Theta$ with high probability. As a result, we can perfectly recover the true feature group partition with high probability.

Given $g,g'\in[G+1]$, for any $j\in\mathcal A_g^*$ and $g'\ne g$, we have
\begin{equation}\label{min_dis}
\begin{aligned}
&\ \left\langle{\mathbf K}_j-\tau\mathbf E_n,\mathbf P_g\right\rangle-\left\langle{\mathbf K}_j-\tau\mathbf E_n,\mathbf P_{g'}\right\rangle\\
=&\ \left\langle{\mathbf K}_j-\tau\mathbf E_n,\mathbf P_g^*-\mathbf P_{g'}^*\right\rangle-\left[\left\langle{\mathbf K}_j-\tau\mathbf E_n,\mathbf P_{g'}-\mathbf P_{g'}^*\right\rangle-\left\langle{\mathbf K}_j-\tau\mathbf E_n,\mathbf P_{g}-\mathbf P_{g}^*\right\rangle\right]\\
\ge&\ \left\langle{\mathbf K}_j-\tau\mathbf E_n,\mathbf P_g^*-\mathbf P_{g'}^*\right\rangle-2r_n,
\end{aligned}
\end{equation}
where the last inequality follows from $\|{\mathbf K}_j-\tau\mathbf E_n\|_{\max}\le 1$ and $\mathbf P\in\bm\Theta$. To obtain the lower bound of $\left\langle{\mathbf K}_j-\tau\mathbf E_n,\mathbf P_g^*-\mathbf P_{g'}^*\right\rangle$, we should discuss ${\mathbf K}_j-\tau\mathbf E_n$ and $\mathbf P_g^*-\mathbf P_{g'}^*$, respectively.

By the definitions of $\mathbf P_g^*$ and $\mathbf P_{g'}^*$, the elements of $\mathbf P_g^*-\mathbf P_{g'}^*$ exhibit three different distribution patterns listed below.
\begin{itemize}
    \item If $x_{i,j}$ and $x_{i',j}$ belong to the same sample cluster or two different clusters in both the $g$-th feature group and the $g'$-th feature group, then $P_{g,ii'}^*-P_{g',ii'}^*=0$;
    \item If $x_{i,j}$ and $x_{i',j}$ belong to the same sample cluster in the $g$-th feature group, while they belong to two different clusters in the $g'$-th feature group, then $P_{g,ii'}^*-P_{g',ii'}^*=1-0=1$;
    \item If $x_{i,j}$ and $x_{i',j}$ belong to two different sample clusters in the $g$-th feature group, while they belong to the same cluster in the $g'$-th feature group, then $P_{g,ii'}^*-P_{g',ii'}^*=0-1=-1$.
\end{itemize}

Firstly, given a pair of $(g,g')$ with $g\ne g'\in[G+1]$ and a pair of $(h,h')\in\mathcal D_1^{gg'}\cap\mathcal D_{\mathrm{I}}^{gg'}$, we consider $\Delta_{j,hh'}^{(gg')}=\mathbb E(K_{j,ii'}|i\in\mathcal H_h^{(gg')},i'\in\mathcal H_{h'}^{(gg')})$ with $j\in\mathcal A_g^*$. Then,
\[
\widehat{\Delta}_{j,hh'}^{(gg')}=\frac{\sum_{i\in\mathcal H_h^{(gg')},i'\in\mathcal H_{h'}^{(gg')}}K_{j,ii'}}{n_h^{(gg')}\cdot n_{h'}^{(gg')}}
\]
is an unbiased estimator for $\Delta_{j,hh'}^{(gg')}$. By Bernstein's inequality for two-sample U-statistics ((E.14) in \cite{srivastava2023robust}), we have
\begin{equation}\label{berstein_onesample}
\mathbb P\left[\widehat{\Delta}_{j,hh'}^{(gg')}-\Delta_{j,hh'}^{(gg')}< -t\right]\le \exp\left\{-\frac{\min\left\{n_h^{(gg')},n_{h'}^{(gg')}\right\}t^2}{c_1'v_{j,hh'}^{(gg')}+c_2't}\right\},
\end{equation}
where $v_{j,hh'}^{(gg')}=\mathrm{Var}(K_{j,ii'}|i\in\mathcal H_h^{(gg')},i'\in\mathcal H_{h'}^{(gg')})$ and $c_1',c_2'>0$ are constants. 
Since there exists some $m\in[M_g]$ such that $\mathcal H_{h}^{(gg')},\mathcal H_{h'}^{(gg')}\subset{\mathcal G}_{g,m}^{*}$, we have
\begin{equation}\label{lower_expectation}
\begin{aligned}
\Delta_{j,hh'}^{(gg')}=&\ \mathbb E\left[\exp\left\{-\frac{(\mu_{m,j}^{(g)}+\epsilon_{i,j}-\mu_{m,j}^{(g)}-\epsilon_{i',j})^2}{2\theta_j^2}\right\}\right]\\
\ge&\ \exp\left\{-\frac{\mathbb E(\epsilon_{i,j}-\epsilon_{i',j})^2}{2\theta_j^2}\right\}\\
\ge&\ \exp\left\{-\frac{\sigma^2}{\theta_j^2}\right\},\\
=&\ \exp\left\{-\frac{1}{\kappa_j^2(\text{SNR}_{\min,j}^{(g)})^2}\right\},\\
\end{aligned}
\end{equation}
where the first inequality follows from Jensen's inequality, and the last inequality follows from the fact that $\epsilon_{ij}$ and $\epsilon_{i'j}$ are independent sub-Gaussian random variables with the same sub-Gaussian norm $\sigma$. 
Moreover, due to $0<K_{j,ii'}\le 1$, we have
\begin{equation}\label{upper_variance}
\begin{aligned}
v_{j,hh'}^{(gg')}=&\ \mathbb E(K_{j,ii'}^2|i\in\mathcal H_h^{(gg')},i'\in\mathcal H_{h'}^{(gg')})-(\Delta_{j,hh'}^{(gg')})^2\le 1-\exp\left\{-\frac{2}{\kappa_j^2(\text{SNR}_{\min,j}^{(g)})^2}\right\}.\\
\end{aligned}
\end{equation}
Let $t_1=\Delta_{j,hh'}^{(gg')}-\zeta_j^{\mathrm{in}}$. By \eqref{lower_expectation} and the definition of $\zeta_j^{\mathrm{in}}$, we have $0<\Delta_{j,hh'}^{(gg')}-\zeta_j^{\mathrm{in}}<\Delta_{j,hh'}^{(gg')}\le 1$. Now we should discuss two cases: $t_1\ge v_{j,hh'}^{(gg')}$ or $t_1< v_{j,hh'}^{(gg')}$. Specifically, if $t_1\ge v_{j,hh'}^{(gg')}$, \eqref{berstein_onesample} can be rewritten as
\begin{equation}\label{berstein_1}
\begin{aligned}
\mathbb P\left[\widehat{\Delta}_{j,hh'}^{(gg')}<\zeta_j^{\mathrm{in}}\big| t_1\ge v_{j,hh'}^{(gg')}\right]\le&\ \exp\left\{-\frac{\min\left\{n_h^{(gg')},n_{h'}^{(gg')}\right\}t_1^2}{c_1'v_{j,hh'}^{(gg')}+c_2't_1}\right\}\\
\le&\ \exp\left\{-\frac{\min\left\{n_h^{(gg')},n_{h'}^{(gg')}\right\}\cdot(\Delta_{j,hh'}^{(gg')}-\zeta_j^{\mathrm{in}})}{c_1'+c_2'}\right\}\\
\le&\ \exp\left\{-C_{1,j}'n_{\min}^{\mathrm{diff}}\right\},\\
\end{aligned}
\end{equation}
where $C_{1,j}'=(1-\kappa_j')\exp\left\{-\frac{1}{\kappa_j^2(\text{SNR}_{\min,j}^{(g)})^2}\right\}/(c_1'+c_2')>0$ is a constant. Besides, if $t_1< v_{j,hh'}^{(gg')}$, under \eqref{berstein_onesample}, \eqref{upper_variance}, we can have
\begin{equation}\label{berstein_2}
\begin{aligned}
\mathbb P\left[\widehat{\Delta}_{j,hh'}^{(gg')}<\zeta_j^{\mathrm{in}}\big |t_1< v_{j,hh'}^{(gg')}\right]\le&\ \exp\left\{-\frac{\min\left\{n_h^{(gg')},n_{h'}^{(gg')}\right\}t_1^2}{c_1'v_{j,hh'}^{(gg')}+c_2't_1}\right\}\\
\le&\ \exp\left\{-\frac{\min\left\{n_h^{(gg')},n_{h'}^{(gg')}\right\}\cdot(\Delta_{j,hh'}^{(gg')}-\zeta_j^{\mathrm{in}})^2}{(c_1'+c_2')\left(1-\exp\left\{-\frac{2}{\kappa_j^2(\mathrm{SNR}_{\min,j}^{(g)})^2}\right\}\right)}\right\}\\
\le&\ \exp\left\{-C_{2,j}'n_{\min}^{\mathrm{diff}}\right\},\\
\end{aligned}
\end{equation}
where $C_{2,j}'=\left[(1-\kappa_j')\exp\left\{-\frac{1}{\kappa_j^2(\text{SNR}_{\min,j}^{(g)})^2}\right\}\right]^2\Big/\left[(c_1'+c_2')\left(1-\exp\left\{-\frac{2}{\kappa_j^2(\mathrm{SNR}_{\min,j}^{(g)})^2}\right\}\right)\right]>0$ is a constant.
Consequently, given $g\ne g'\in[G+1]$, $j\in\mathcal A_g^*$ and a pair of $(h,h')\in\mathcal D_1^{gg'}$,
\[
\begin{aligned}
\mathbb P\left[\widehat{\Delta}_{j,hh'}^{(gg')}<\zeta_j^{\mathrm{in}}\right]=&\ \mathbb P\left[\widehat{\Delta}_{j,hh'}^{(gg')}<\zeta_j^{\mathrm{in}}\big |t_1< v_{j,hh'}^{(gg')}\right]\mathbb P\left[t_1< v_{j,hh'}^{(gg')}\right]\\
&\ +\mathbb P\left[\widehat{\Delta}_{j,hh'}^{(gg')}<\zeta_j^{\mathrm{in}}\big |t_1\ge v_{j,hh'}^{(gg')}\right]\mathbb P\left[t_1\ge v_{j,hh'}^{(gg')}\right]\\
\le&\ \mathbb P\left[\widehat{\Delta}_{j,hh'}^{(gg')}<\zeta_j^{\mathrm{in}}\big |t_1< v_{j,hh'}^{(gg')}\right]+\mathbb P\left[\widehat{\Delta}_{j,hh'}^{(gg')}<\zeta_j^{\mathrm{in}}\big |t_1\ge v_{j,hh'}^{(gg')}\right]\\
\le&\ \exp\left\{-C_{2,j}'n_{\min}^{\mathrm{diff}}\right\}+\exp\left\{-C_{1,j}'n_{\min}^{\mathrm{diff}}\right\}.
\end{aligned}
\]
As a result, given $g\in[G+1]$ and $j\in\mathcal A_g^*$,
\[
\mathbb P\left[\min_{g'\ne g\in[G+1],(h,h')\in\mathcal D_1^{gg'}}\widehat{\Delta}_{j,hh'}^{(gg')}\ge\zeta_j^{\mathrm{in}}\right]\ge 1-G|\mathcal D_{1}|\times\left[\exp\left\{-C_{1,j}'n_{\min}^{\mathrm{diff}}\right\}+\exp\left\{-C_{2,j}'n_{\min}^{\mathrm{diff}}\right\}\right],
\]
where $|\mathcal D_{1}|:=\max_{g'\ne g\in[G+1]}|\mathcal D_{1}^{gg'}|$. Let $\zeta_{\min}^{\mathrm{in}}:=\min_{j\in[p]}\zeta_j^{\mathrm{in}}$ and $C_{12}':=\min_{j\in[p]}\min\{C_{1,j}',C_{2,j}'\}$,
then
\begin{equation}\label{eq1}
\mathbb P\left[\min_{g\in[G+1],j\in\mathcal A_g^*}\min_{g'\ne g\in[G+1],(h,h')\in\mathcal D_1^{gg'}}\widehat{\Delta}_{j,hh'}^{(gg')}\ge\zeta_{\min}^{\mathrm{in}}\right]\ge 1-2pG|\mathcal D_{1}|\times\left[\exp\left\{-C_{12}'n_{\min}^{\mathrm{diff}}\right\}\right].
\end{equation}

Secondly, given a pair of $(g,g')$ with $g\ne g'\in[G+1]$ and a pair of $(h,h')\in\mathcal D_{-1}^{gg'}$, we also consider $\Delta_{j,hh'}^{(gg')}=\mathbb E(K_{j,ii'}|i\in\mathcal H_h^{(gg')},i'\in\mathcal H_{h'}^{(gg')})$ with $j\in\mathcal A_g^*$. Then,
\[
\widehat{\Delta}_{j,hh'}^{(gg')}=\frac{\sum_{i\in\mathcal H_h^{(gg')},i'\in\mathcal H_{h'}^{(gg')}}K_{j,ii'}}{n_h^{(gg')}\cdot n_{h'}^{(gg')}}
\]
is an unbiased estimator for $\Delta_{j,hh'}^{(gg')}$. In this case, we can find that there exist $m\ne m'\in[M_g]$ such that $\mathcal H_{h}^{(gg')}\subset{\mathcal G}_{g,m}^{*}$ and $\mathcal H_{h'}^{(gg')}\subset{\mathcal G}_{g,m'}^{*}$. 
By arguments analogous to those used in deriving \eqref{eqa1} -- \eqref{eqa5}, we have
\begin{equation}\label{lower_expectation2}
\begin{aligned}
\Delta_{j,hh'}^{(gg')}=&\ \mathbb E\left[\exp\left\{-\frac{(\mu_{m,j}^{(g)}+\epsilon_{i,j}-\mu_{m',j}^{(g)}-\epsilon_{i',j})^2}{2\theta_j^2}\right\}\right]\\
\le&\ \exp\left\{-\frac{(d_{\min,j}^{(g)})^2}{8\theta_j^2}\right\}+2\exp\left\{-\frac{(d_{\min,j}^{(g)})^2}{16\sigma^2}\right\}\\
=&\ \exp\left\{-\frac{1}{8\kappa_j^2}\right\}+2\exp\left\{-\frac{(\text{SNR}_{\min,j}^{(g)})^2}{16}\right\}.
\end{aligned}
\end{equation}
Similarly, we can show that 
\begin{equation}\label{upper_variance2}
v_{j,hh'}^{(gg')}\le\mathbb E(K_{j,ii'}^2|i\in\mathcal H_h^{(gg')},i'\in\mathcal H_{h'}^{(gg')})\le\exp\left\{-\frac{(d_{\min,j}^{(g)})^2}{4\theta_j^2}\right\}+2\exp\left\{-\frac{(d_{\min,j}^{(g)})^2}{16\sigma^2}\right\}.
\end{equation}
Let $t_2=\zeta_j^{\mathrm{out}}-\Delta_{j,hh'}^{(gg')}$. By \eqref{lower_expectation2} and the definition of $\zeta_j^{\mathrm{out}}$, we have $t_2>0$. Now we should discuss two cases: $t_2\ge v_{j,hh'}^{(gg')}$ or $t_2< v_{j,hh'}^{(gg')}$. Specifically, if $t_2\ge v_{j,hh'}^{(gg')}$, by \eqref{berstein_onesample}, we can have
\begin{equation}\label{berstein_3}
\begin{aligned}
\mathbb P\left[\widehat{\Delta}_{hh'}^{(j)}>\zeta_j^{\mathrm{out}}\big|t_2\ge v_{j,hh'}^{(gg')}\right]\le&\ \exp\left\{-\frac{\min\left\{n_h^{(gg')},n_{h'}^{(gg')}\right\}t_2^2}{c_1'v_{j,hh'}^{(gg')}+c_2't_2}\right\}\\
\le&\ \exp\left\{-\frac{\min\left\{n_h^{(gg')},n_{h'}^{(gg')}\right\}\cdot(\zeta_j^{\mathrm{out}}-\Delta_{j,hh'}^{(gg')})}{c_1'+c_2'}\right\}\\
\le&\ \exp\left\{-C_{3,j}'n_{\min}^{\mathrm{diff}}\right\},\\
\end{aligned}
\end{equation}
where $C_{3,j}'=(\zeta_j^{\mathrm{out}}-\Delta_{j,hh'}^{(gg')})/(c_1'+c_2')>0$ is a constant. Besides, if $t_2< v_{j,hh'}^{(gg')}$, under \eqref{berstein_onesample} and \eqref{upper_variance2}, we can have
\begin{equation}\label{berstein_4}
\begin{aligned}
&\ \mathbb P\left[\widehat{\Delta}_{j,hh'}^{(gg')}>\zeta_j^{\mathrm{out}}\big|t_2< v_{j,hh'}^{(gg')}\right]\\
\le&\ \exp\left\{-\frac{\min\left\{n_h^{(gg')},n_{h'}^{(gg')}\right\}t_2^2}{c_1'v_{j,hh'}^{(gg')}+c_2't_2}\right\}\\
\le&\ \exp\left\{-\frac{\min\left\{n_h^{(gg')},n_{h'}^{(gg')}\right\}\cdot(\zeta_j^{\mathrm{out}}-\Delta_{j,hh'}^{(gg')})^2}{(c_1'+c_2')\left(\exp\left\{-\frac{(d_{\min,j}^{(g)})^2}{4\theta_j^2}\right\}+2\exp\left\{-\frac{(d_{\min,j}^{(g)})^2}{16\sigma^2}\right\}\right)}\right\}\\
\le&\ \exp\left\{-C_{4,j}'n_{\min}^{\mathrm{diff}}\right\},\\
\end{aligned}
\end{equation}
where $C_{4,j}'=(\zeta_j^{\mathrm{out}}-\Delta_{j,hh'}^{(gg')})^2/\left[(c_1'+c_2')\left(\exp\left\{-\frac{(d_{\min,j}^{(g)})^2}{4\theta_j^2}\right\}+2\exp\left\{-\frac{(d_{\min,j}^{(g)})^2}{16\sigma^2}\right\}\right)\right]>0$ is a constant. 
Consequently, given $g\ne g'\in[G+1]$, $j\in\mathcal A_g^*$ and a pair of $(h,h')\in\mathcal D_{-1}^{gg'}$,
\[
\begin{aligned}
\mathbb P\left[\widehat{\Delta}_{j,hh'}^{(gg')}>\zeta_j^{\mathrm{out}}\right]=&\ \mathbb P\left[\widehat{\Delta}_{j,hh'}^{(gg')}>\zeta_j^{\mathrm{out}}\big |t_2< v_{j,hh'}^{(gg')}\right]\mathbb P\left[t_2< v_{j,hh'}^{(gg')}\right]\\
&\ +\mathbb P\left[\widehat{\Delta}_{j,hh'}^{(gg')}>\zeta_j^{\mathrm{out}}\big |t_2\ge v_{j,hh'}^{(gg')}\right]\mathbb P\left[t_2\ge v_{j,hh'}^{(gg')}\right]\\
\le&\ \mathbb P\left[\widehat{\Delta}_{j,hh'}^{(gg')}>\zeta_j^{\mathrm{out}}\big |t_2< v_{j,hh'}^{(gg')}\right]+\mathbb P\left[\widehat{\Delta}_{j,hh'}^{(gg')}>\zeta_j^{\mathrm{out}}\big |t_2\ge v_{j,hh'}^{(gg')}\right]\\
\le&\ \exp\left\{-C_{4,j}'n_{\min}^{\mathrm{diff}}\right\}+\exp\left\{-C_{3,j}'n_{\min}^{\mathrm{diff}}\right\}.
\end{aligned}
\]

As a result, given $g\in[G+1]$ and $j\in\mathcal A_g^*$,
\[
\mathbb P\left[\max_{g'\ne g\in[G+1],(h,h')\in\mathcal D_{-1}^{gg'}}\widehat{\Delta}_{j,hh'}^{(gg')}\le\zeta_j^{\mathrm{out}}\right]\ge 1-G|\mathcal D_{-1}|\times\left[\exp\left\{-C_{3,j}'n_{\min}^{\mathrm{diff}}\right\}+\exp\left\{-C_{4,j}'n_{\min}^{\mathrm{diff}}\right\}\right],
\]
where $|\mathcal D_{-1}|:=\max_{g\in[G+1]}\max_{g'\ne g\in[G+1]}|\mathcal D_{-1}^{gg'}|$. Let $\zeta_{\max}^{\mathrm{out}}:=\max_{j\in[p]}\zeta_j^{\mathrm{out}}$ and $C_{34}':=\min_{j\in[p]}\min\{C_{3,j}',C_{4,j}'\}$,
then
\begin{equation}\label{eq2}
\mathbb P\left[\max_{g\in[G+1],j\in\mathcal A_g^*}\max_{g'\ne g\in[G+1],(h,h')\in\mathcal D_1^{gg'}}\widehat{\Delta}_{j,hh'}^{(gg')}\le\zeta_{\max}^{\mathrm{out}}\right]\ge 1-2pG|\mathcal D_{-1}|\times\left[\exp\left\{-C_{34}'n_{\min}^{\mathrm{diff}}\right\}\right].
\end{equation}

Therefore, combining \eqref{eq1}, \eqref{eq2}, and taking $C'=\min\{C_{12}',C_{34}'\}$, then with probability at least $1-2pG(|\mathcal D_1|+|\mathcal D_{-1}|)\exp\left\{-C'n_{\min}^{\mathrm{diff}}\right\}$,
\[
\begin{aligned}
\inf_{g\ne g'\in[G+1],j\in\mathcal A_g^*}\left\langle{\mathbf K}_j-\tau\mathbf E_n,\mathbf P_g^*-\mathbf P_{g'}^*\right\rangle\ge&\ \min\{\zeta_{\min}^{\mathrm{in}}-\tau,\tau-\zeta_{\max}^{\mathrm{out}}\}\sum_{(h,h')\in\mathcal D_{\mathrm{I}}^{(gg')}}\left\|\mathbf P_{g,hh'}^*-\mathbf P_{g',hh'}^*\right\|_1\\
&\quad -\max\{\tau,1-\tau\}\sum_{(h,h')\in\mathcal D_{\mathrm{II}}^{(gg')}}\left\|\mathbf P_{g,hh'}^*-\mathbf P_{g',hh'}^*\right\|_1
\end{aligned}
\]
Consequently, if $\Lambda_{gg'}> 2r_n$ hold for any $g\ne g'\in[G+1]$. Then, for any $g\ne g'\in[G+1]$ and $j\in\mathcal A_g^*$, assigning $\mathbf K_j$ from $\mathcal A_g^*$ to $\mathcal A_{g'}^*$ is certain to decrease $\mathcal L(\mathcal A,\mathcal P)$ with high probability, which concludes the proof.
$\hfill\square$

\newpage
\noindent\textbf{Proof of Theorem 3.} To show the result in Theorem 3, if suffices to verify the conditions in Lemma \ref{lemma4} hold with high probability. 
Recall Definitions \ref{def3} and \ref{def4}.
Specifically, for $g\in[G]$, it is easy to show that $\mathbf P_g^*\in\mathbb M_{n,n}^{M_g}$, $n_{m_g}(\mathbf P_g^*)=n_{m_g}^{(g)}$ and $\delta_{m_g}(\mathbf P_g^*)\ge \sqrt{2n_{\min}}$. Let $\widehat{\mathbf P}_g^{(M_g)}$ be the $M_g$-truncated eigenvalue decomposition (EVD) of $\widehat{\mathbf P}_g$ and $\lambda_k(\widehat{\mathbf P}_g)$ be the $k$-th largest eigenvalue of $\widehat{\mathbf P}_g$. Note that, for squared matrix, the $k$-th singular value is the corresponding $k$-th absolute eigenvalue, then we have
\[
\begin{aligned}
\left\|\widehat{\mathbf P}_g^{(M_g)}-\mathbf P_g^*\right\|_2\le&\ \left\|\widehat{\mathbf P}_g^{(M_g)}-\widehat{\mathbf P}_g\right\|_2+\left\|\widehat{\mathbf P}_g-\mathbf P_g^*\right\|_2\\
=&\ \left|\lambda_{M_g+1}(\widehat{\mathbf P}_g)\right|+\left\|\widehat{\mathbf P}_g-\mathbf P_g^*\right\|_2\\
\le&\ 2\left\|\widehat{\mathbf P}_g-\mathbf P_g^*\right\|_2,\\
\end{aligned}
\]
where the equality follows from (33) in \cite{zhou2019analysis} and the last inequality follows from Weyl's theorem on the perturbation of singular values, that is, $\big||\lambda_i(\widehat{\mathbf P}_g)|-|\lambda_i(\mathbf P_g^*)|\big|\le \|\widehat{\mathbf P}_g-\mathbf P_g^*\|_2$ as well as $\lambda_{M_g+1}(\widehat{\mathbf P}_g^*)=0$. Since $\widehat{\mathbf P}_g^{(M_g)}-\mathbf P_g^*$ is of rank at most $2M_g$, we can have
\[
\begin{aligned}
\left\|\widehat{\mathbf P}_g^{(M_g)}-\mathbf P_g^*\right\|_F\le&\ \sqrt{2M_g}\left\|\widehat{\mathbf P}_g^{(M_g)}-\mathbf P_g^*\right\|_2\\
\le&\ 2\sqrt{2M_g}\left\|\widehat{\mathbf P}_g-\mathbf P_g^*\right\|_2\\
\le&\ 2\sqrt{2M_g}\left\|\widehat{\mathbf P}_g-\mathbf P_g^*\right\|_F\\
=&\ 2\sqrt{2M_g}\left\|\widehat{\mathbf P}_g-\mathbf P_g^*\right\|_1^{1/2}\\
\le&\ 2\sqrt{2M_g}r_n^{1/2}
\end{aligned}
\]
with probability at least $1-2GM_{\max}/n_{\min}-2pG(|\mathcal D_1|+|\mathcal D_{-1}|)\cdot\exp\{-C'n_{\min}^{\mathrm{diff}}\}$, where the equality holds due to the fact that all elements of $\widehat{\mathbf P}_g,\mathbf P_g^*$ are 0 or 1, and the last inequality follows from Theorems \ref{thm1} and \ref{thm2}. 

Now we are going to check Condition (b) in Lemma \ref{lemma4}. Let $c_{m_g}=1/2$ for each $m_g\in[M_g]$, since $n\rightarrow\infty$ and $s_{\min}\rightarrow\infty$, with probability at least $1-2GM_{\max}/n_{\min}-2pG(|\mathcal D_1|+|\mathcal D_{-1}|)\cdot\exp\{-C'n_{\min}^{\mathrm{diff}}\}$, we have
\[
\frac{c_{m_g}^{-2}(1+\omega)^2\zeta^2}{[\delta_{m_g}(\mathbf P_g^*)]^2n_{m_g}(\mathbf P_g^*)}\le \frac{4(1+\omega)^2\cdot 8M_gr_n}{2n_{\min}^2}\le 16M_{\max}(1+\omega)^2r_n/n_{\min}^2\le 1
\]
for sufficiently large $n$ and $s_{\min}$. Accordingly, with probability at least $1-2GM_{\max}/n_{\min}-2pG(|\mathcal D_1|+|\mathcal D_{-1}|)\cdot\exp\{-C'n_{\min}^{\mathrm{diff}}\}$, we have
\[
\sup_{g\in[G]}\overline{\mathrm{Mis}}\left(\mathbf P_g^*;\mathcal P_{\omega}(\widehat{\mathbf P}_g^{(M_g)})\right)\le 16M_{\max}(1+\omega)^2C''\cdot\max\left\{\exp[-\kappa s_{\min}],\frac{\log n_{\min}}{n_{\min}}\right\},
\]
where $C''>0$ is a constant.

\section{Additional remarks for Section 3.2}

Figure \ref{fig:toy1} is provided as a toy example to help readers understand the definitions involved in measuring the discrepancy between two clustering matrices.
Figures \ref{fig:toy1}(1) and \ref{fig:toy1}(2) illustrate the true clustering matrices associated with two distinct clustering structures. 
An occupied entry at position $(i,j)$ indicates that observations $i$ and $j$ are assigned to the same cluster, while an empty entry indicates that they belong to different clusters. 
Figure \ref{fig:toy1}(3) visualizes the entrywise difference between the two clustering matrices. 
Red entries correspond to locations where the difference equals 1, green entries correspond to locations where the difference equals -1, and blank entries indicate no discrepancy between the two clustering matrices (i.e., a difference of 0).
For ease of theoretical analysis, we can permute the rows and columns of the clustering matrix in Figure \ref{fig:toy1}(2), leading to a rearranged clustering matrix in Figure \ref{fig:toy1}(5). As a result, a explicit separation matrix with block-wise 1's, -1's or 0's are obtained. 
In Figure \ref{fig:toy1}(6), $h_{gg'}=4$, the red-dot region represents a subset of $\mathcal D_1^{gg'}$, whereas the green-dot region represents a subset of $\mathcal D_{-1}^{gg'}$.

}

\begin{figure}[H]
    \centering
\begin{tikzpicture}
    \node at (0,8) {
    \begin{tikzpicture}[scale=0.8] 
    \draw[step=0.2cm, gray, very thin] (0,0) grid (6,6);

    \foreach \i in {1,...,30} {
        \pgfmathsetmacro{\yy}{6.1-0.2*\i}
        \node[left,font=\fontsize{4pt}{4pt}\selectfont] at (-0.05,\yy) {$\i$};
    }
    
    \foreach \x in {0, 0.2, ..., 1.8} {
        \foreach \y in {4, 4.2, ..., 5.8} {
           \fill[blue] (\x+0.1,\y+0.1) circle(0.05); 
        }
    }
    \foreach \x in {2, 2.2, ..., 5.8} {
        \foreach \y in {0, 0.2, ..., 3.8} {
          \fill[orange] (\x+0.1,\y+0.1) circle(0.05);
        }
    }
    \node[below,font=\small] at (3,-0.35) {(1)};
    \end{tikzpicture}
    };

    \node at (5.5,8) {
    \begin{tikzpicture}[scale=0.8] 
    \draw[step=0.2cm, gray, very thin] (0,0) grid (6,6);

    \foreach \i in {1,...,30} {
        \pgfmathsetmacro{\yy}{6.1-0.2*\i}
        \node[left,font=\fontsize{4pt}{4pt}\selectfont] at (-0.05,\yy) {$\i$};
    }
    
    \foreach \x in {0, 0.4, ..., 5.6} {
        \foreach \y in {0.2, 0.6, ..., 6.2} {
           \fill[blue] (\x+0.1,\y+0.1) circle(0.05); 
        }
    }
    \foreach \x in {0.2, 0.6, ..., 6.2} {
        \foreach \y in {0, 0.4, ..., 5.6} {
           \fill[orange] (\x+0.1,\y+0.1) circle(0.05); 
        }
    }
    \node[below,font=\small] at (3,-0.35) {(2)};
    \end{tikzpicture}
    };

    \node at (11,8) {
\begin{tikzpicture}[scale=0.8] 
    \draw[step=0.2cm, gray, very thin] (0,0) grid (6,6);

    \foreach \i in {1,...,30} {
        \pgfmathsetmacro{\yy}{6.1-0.2*\i}
        \node[left,font=\fontsize{4pt}{4pt}\selectfont] at (-0.05,\yy) {$\i$};
    }

    \foreach \c in {1,...,30} {
        \foreach \r in {1,...,30} {

            \pgfmathsetmacro{\xx}{0.2*(\c-1)}
            \pgfmathsetmacro{\yy}{0.2*(\r-1)}

            \pgfmathtruncatemacro{\first}{
                ((\c <= 10 && \r >= 21) || (\c >= 11 && \r <= 20)) ? 1 : 0
            }

            \pgfmathtruncatemacro{\second}{
                mod(\c+\r,2) == 1 ? 1 : 0
            }

            \ifnum\first=1
                \ifnum\second=0
                    \fill[red] (\xx+0.1,\yy+0.1) circle(0.05);
                \fi
            \fi

            \ifnum\first=0
                \ifnum\second=1
                    \fill[green] (\xx+0.1,\yy+0.1) circle(0.05);
                \fi
            \fi

        }
    }
\node[below,font=\small] at (3,-0.35) {(3)};
\end{tikzpicture}
};

    \node at (0,0) {
    \begin{tikzpicture}[scale=0.8] 
    \draw[step=0.2cm, gray, very thin] (0,0) grid (6,6);

    \foreach \lab [count=\i from 1] in {1,3,5,7,9,2,4,6,8,10,12,14,16,18,20,22,24,26,28,30,11,13,15,17,19,21,23,25,27,29} {
        \pgfmathsetmacro{\yy}{6.1-0.2*\i}
        \node[left,font=\fontsize{4pt}{4pt}\selectfont] at (-0.05,\yy) {$\lab$};
    }
    
    \foreach \x in {0, 0.2, ..., 1.8} {
        \foreach \y in {4, 4.2, ..., 5.8} {
           \fill[blue] (\x+0.1,\y+0.1) circle(0.05); 
        }
    }
    \foreach \x in {2, 2.2, ..., 5.8} {
        \foreach \y in {0, 0.2, ..., 3.8} {
          \fill[orange] (\x+0.1,\y+0.1) circle(0.05);
        }
    }
\node[below,font=\small] at (3,-0.35) {(4)};
    \end{tikzpicture}
    };

    \node at (5.5,0) {
    \begin{tikzpicture}[scale=0.8] 
    \draw[step=0.2cm, gray, very thin] (0,0) grid (6,6);

    \foreach \lab [count=\i from 1] in {1,3,5,7,9,2,4,6,8,10,12,14,16,18,20,22,24,26,28,30,11,13,15,17,19,21,23,25,27,29} {
        \pgfmathsetmacro{\yy}{6.1-0.2*\i}
        \node[left,font=\fontsize{4pt}{4pt}\selectfont] at (-0.05,\yy) {$\lab$};
    }
    
    \foreach \x in {0, 0.2, ..., 0.8} {
        \foreach \y in {5, 5.2, ..., 5.8} {
           \fill[blue] (\x+0.1,\y+0.1) circle(0.05); 
        }
    }
    \foreach \x in {4, 4.2, ..., 5.8} {
        \foreach \y in {5, 5.2, ..., 5.8} {
           \fill[blue] (\x+0.1,\y+0.1) circle(0.05); 
        }
    }
    \foreach \x in {1, 1.2, ..., 3.8} {
        \foreach \y in {2, 2.2, ..., 4.8} {
           \fill[orange] (\x+0.1,\y+0.1) circle(0.05); 
        }
    }
        \foreach \x in {0, 0.2, ..., 0.8} {
        \foreach \y in {0, 0.2, ..., 1.8} {
          \fill[blue] (\x+0.1,\y+0.1) circle(0.05);
        }
    }
    \foreach \x in {4, 4.2, ..., 5.8} {
        \foreach \y in {0, 0.2, ..., 1.8} {
          \fill[blue] (\x+0.1,\y+0.1) circle(0.05);
        }
    }
\node[below,font=\small] at (3,-0.35) {(5)};
    \end{tikzpicture}
    };

    \node at (11,0) {
    \begin{tikzpicture}[scale=0.8] 
    \draw[step=0.2cm, gray, very thin] (0,0) grid (6,6);

    \foreach \lab [count=\i from 1] in {1,3,5,7,9,2,4,6,8,10,12,14,16,18,20,22,24,26,28,30,11,13,15,17,19,21,23,25,27,29} {
        \pgfmathsetmacro{\yy}{6.1-0.2*\i}
        \node[left,font=\fontsize{4pt}{4pt}\selectfont] at (-0.05,\yy) {$\lab$};
    }
    
    \foreach \x in {1, 1.2, ..., 1.8} {
        \foreach \y in {5, 5.2, ..., 5.8} {
           \fill[red] (\x+0.1,\y+0.1) circle(0.05); 
        }
    }
    \foreach \x in {0, 0.2, ..., 0.8} {
        \foreach \y in {4, 4.2, ..., 4.8} {
           \fill[red] (\x+0.1,\y+0.1) circle(0.05); 
        }
    }
    \foreach \x in {2, 2.2, ..., 3.8} {
        \foreach \y in {4, 4.2, ..., 4.8} {
          \fill[green] (\x+0.1,\y+0.1) circle(0.05);
        }
    }
    \foreach \x in {4, 4.2, ..., 5.8} {
        \foreach \y in {5, 5.2, ..., 5.8} {
          \fill[green] (\x+0.1,\y+0.1) circle(0.05);
        }
    }
    \foreach \x in {1, 1.2, ..., 1.8} {
        \foreach \y in {2, 2.2, ..., 3.8} {
          \fill[green] (\x+0.1,\y+0.1) circle(0.05);
        }
    }
    \foreach \x in {0, 0.2, ..., 0.8} {
        \foreach \y in {0, 0.2, ..., 1.8} {
          \fill[green] (\x+0.1,\y+0.1) circle(0.05);
        }
    }
    \foreach \x in {2, 2.2, ..., 3.8} {
        \foreach \y in {0, 0.2, ..., 1.8} {
          \fill[red] (\x+0.1,\y+0.1) circle(0.05);
        }
    }
    \foreach \x in {4, 4.2, ..., 5.8} {
        \foreach \y in {2, 2.2, ..., 3.8} {
          \fill[red] (\x+0.1,\y+0.1) circle(0.05);
        }
    }
    \node[below,font=\small] at (3,-0.35) {(6)};
    \end{tikzpicture}
    };
\end{tikzpicture}

    \caption{A toy example illustrating the key definitions in the theoretical analysis.}
    \label{fig:toy1}
\end{figure}

\addtolength{\textheight}{-.2in}%

\end{document}